\RequirePackage{fix-cm}
\documentclass[twocolumn]{svjour3}          
\smartqed  
\usepackage{hyperref}
\usepackage{textcomp}
\usepackage{graphicx}
\usepackage{times,amsmath,epsfig}
\usepackage{mathtools}
\usepackage{amsmath}
\usepackage{amssymb}
\usepackage{balance}  
\usepackage{comment}
\usepackage{footmisc}
\usepackage{url}
\usepackage[noend,linesnumbered,algoruled,boxed,lined]{algorithm2e}
\usepackage[usenames,dvipsnames,svgnames,table]{xcolor}

\newcommand{\qedsymbol}{$\blacksquare$}

\makeatletter
\newcommand*{\inlineequation}[2][]{%
	\begingroup
	\refstepcounter{equation}%
	\ifx\\#1\\%
	\else
	\label{#1}%
	\fi
	\binoppenalty=10000 %
	\ensuremath{%
		#2%
	}%
	~\@eqnnum
	\endgroup
}
\makeatother

\begin{document}

\title{Scalable Data Series Subsequence Matching with ULISSE}




\author{Michele Linardi         \and
        Themis Palpanas 
}


\institute{Michele Linardi \at
             LIPADE, Universit{\'e} de Paris\\
              \email{michele.linardi@parisdescartes.fr}           
           \and
           Themis Palpanas \at
               LIPADE, Universit{\'e} de Paris\\
              \email{themis@mi.parisdescartes.fr}
}

\date{Received: date / Accepted: date}

\maketitle

\begin{abstract}
Data series similarity search is an important operation and at the core of several analysis tasks and applications related to data series collections.
Despite the fact that data series indexes enable fast similarity search, all existing indexes can only answer queries of a single length (fixed at index construction time), which is a severe limitation.
In this work, we propose \textit{ULISSE}, the first data series index structure designed for answering similarity search queries of \emph{variable length} (within some range). 
Our contribution is two-fold.
First, we introduce a novel representation technique, which effectively and succinctly summarizes multiple sequences of different length. 
Based on the proposed index, we describe efficient algorithms for approximate and exact similarity search, combining disk based index visits and in-memory sequential scans. 
Our approach supports non Z-normalized and Z-normalized sequences, and can be used with no changes with both Euclidean Distance and Dynamic Time Warping, for answering both \emph{k-NN} and $\epsilon$-range queries.
We experimentally evaluate our approach using several synthetic and real datasets. 
The results show that \textit{ULISSE} is several times, and up to orders of magnitude more efficient in terms of both space and time cost, when compared to competing approaches.
(Paper published in VLDBJ 2020)
\end{abstract}

\section{Introduction}

\noindent{\bf Motivation.}
Data sequences are one of the most common data types, and they are present in almost every scientific and social domain (example application domains include meteorology, astronomy, chemistry, medicine, neuroscience, finance, agriculture, entomology, sociology, smart cities, marketing, operation health monitoring, human action recognition and others)~\cite{KashinoSM99,percomJournal,Shasha99,HuijseEPPZ14,DBLP:journals/sigmod/Palpanas15}.
This makes data series a data type of particular importance.

Informally, a data series (a.k.a data sequence, or time series) is defined as an ordered sequence of points, each one associated with a position and a corresponding value\footnote{If the dimension that imposes the ordering of the sequence is time then we talk about time series. Though, a series can also be defined over other measures (e.g., angle in radial profiles in astronomy, mass in mass spectroscopy in physics, 
	etc.). We use the terms \emph{data series}, \emph{time series}, and \emph{sequence} interchangeably.}.
Recent advances in sensing, networking, data processing and storage technologies have significantly facilitated the processes of generating and collecting tremendous amounts of data sequences from a wide variety of domains at extremely high rates and volumes. 

The \emph{SENTINEL-2} mission~\cite{SENTINEL-2} conducted by the European Space Agency (ESA) represents such an example of massive data series collection. 
The two satellites of this mission continuously capture multi-spectral images, designed to give a full picture of earth's surface every five days at a resolution of \textit{10m}, resulting in over five trillion different data series. 
Such recordings will help monitor at fine granularity the evolution of the properties of the surface of the earth, and benefit applications such as land management, agriculture and forestry, disaster control, humanitarian relief operations, risk mapping and security concerns.

\noindent{\bf Data series analytics.}
Once the data series have been collected, the domain experts face the arduous tasks of processing and analyzing them~\cite{KostasThemisTalkICDE,Palpanas2019,DBLP:journals/dagstuhl-reports/BagnallCPZ19} in order to identify patterns, gain insights, detect abnormalities, and extract useful knowledge.
Critical part of this process is the data series similarity search operation, which lies at the core of several analysis and machine learning algorithms (e.g., clustering~\cite{Niennattrakul:2007:CMT:1262690.1262948}, classification~\cite{Lines2015}, outliers~\cite{DBLP:conf/edbt/Senin0WOGBCF15,norma,series2graph}, and others).

However, similarity search in very large data series collections is notoriously challenging~\cite{DBLP:conf/sigmod/ZoumpatianosIP14,DBLP:conf/sofsem/Palpanas16,DBLP:conf/ieeehpcs/Palpanas17,DBLP:conf/ieeehpcs/Palpanas17,DBLP:conf/edbt/GogolouTPB19,conf/sigmod/gogolou20,DBLP:journals/pvldb/EchihabiZPB18,DBLP:journals/pvldb/EchihabiZPB19,}, due to the high dimensionality (length) of the data series.
In order to address this problem, a significant amount of effort has been dedicated by the data management research community to data series indexing techniques~\cite{evolutionofanindex,DBLP:journals/pvldb/EchihabiZPB18,DBLP:journals/pvldb/EchihabiZPB19}, which lead to fast and scalable similarity search~\cite{Faloutsos1994,Rafiei1998,DBLP:conf/vldb/PalpanasCGKZ04,Assent2008,shieh2008sax,DBLP:journals/kais/KadiyalaS08,DBLP:journals/pvldb/WangWPWH13,DBLP:journals/kais/CamerraSPRK14,DBLP:journals/pvldb/DallachiesaPI14,ZoumpatianosIP15,DBLP:journals/vldb/ZoumpatianosIP16,DBLP:conf/icdm/YagoubiAMP17,dpisaxjournal,DBLP:conf/bigdataconf/PengFP18,parisplus,peng2020messi,seriesgoneparallel,coconut,coconutpalm,DBLP:journals/vldb/KondylakisDZP19}.

\noindent{\bf Predefined constraints.}
%
Despite the effectiveness and benefits of the proposed indexing techniques, which have enabled and powered many applications over the years, they are restricted in different ways: either they only support similarity search with queries of a fixed size, or they do not offer a scalable solution. The solutions working for a fixed length, require that this length is chosen at index construction time (it should be the same as the length of the series in the index).

Evidently, this is a constraint that penalizes the flexibility needed by analysts, who often times need to analyze patterns of slightly different lengths (within a given data series collection)~\cite{DBLP:journals/kais/KadiyalaS08,914838,DBLP:conf/kdd/RakthanmanonCMBWZZK12,VALMOD,VALMOD_DEMO,valmodjournal,variablelengthanalytics,Palpanas2019,DBLP:journals/dagstuhl-reports/BagnallCPZ19}. 
This is true for several applications.
For example, in the \emph{SENTINEL-2} mission data, oceanographers are interested in searching for similar coral bleaching patterns\footnote{\scriptsize\url{http://www.esa.int/Our_Activities/Observing_the_Earth/
		%http://www.esa.int/Our_Activities/Observing_the_Earth/Copernicus/Sentinel-2
}} of different lengths; at Airbus\footnote{\scriptsize\url{http://www.airbus.com/}}
engineers need to perform similarity search queries for patterns of variable length when studying aircraft takeoffs and landings~\cite{Airbus}; and in neuroscience, analysts need to search in Electroencephalogram (EEG) recordings for Cyclic Alternating Patterns (CAP) of different lengths (duration), in order to get insights about brain activity during sleep~\cite{ROSA1999585}. 
In these applications, we have datasets with a very large number of fixed length data series, on which analysts need to perform a large number of ad hoc similarity 
queries of (slightly) different lengths (as shown in Figure~\ref{simSearchVL}). 

A straightforward solution for answering such queries would be to use one of the available indexing techniques. 
However, in order to support (exact) results for variable-length similarity search, we would need to (i) create several distinct indexes, one for each possible query length; and (ii) for each one of these indexes, index all overlapping subsequences (using a sliding window).
We illustrate this in Figure~\ref{simSearchVL}, where we depict two similarity search queries of different lengths ($\ell$ and $\ell'$). 
Given a data series from the collection, $D_i$ (shown in black), we draw in red the subsequences that we need to compare to each query in order to compute the exact answer. 
Using an indexing technique implies inserting all the subsequences in the index: since we want to answer queries of two different lengths, we are obliged to use two distinct indexes.

\begin{figure}[tb]
	\centering
	\includegraphics[trim={2cm 13cm 15cm 3cm},scale=0.56]{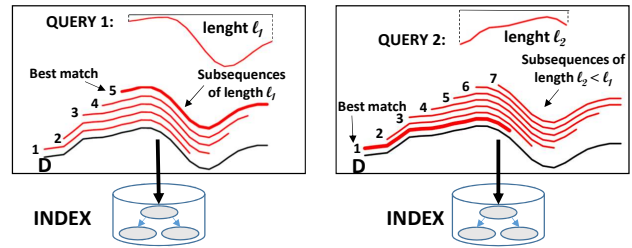}
	\caption{Indexing for supporting queries of 2 different lengths.}
	\label{simSearchVL}
\end{figure}

Nevertheless, this solution is prohibitively expensive, in both space and time. 
Space complexity is increased, since we need to index a large number of subsequences for each one of the supported query lengths: given a data series collection $C={D^{1},...,D^{|C|}}$
and a query length range $[\ell_{min},\ell_{max}]$, the number of subsequences we would normally have to examine (and index) is: 

$S_{\ell_{min},\ell_{max}} = \sum_{\ell=1}^{(\ell_{max}-\ell_{min})+1}\sum_{i=1}^{|C|}(|D^{i}|-(\ell-1))$.
Figure~\ref{spaceExplosionVL} shows how quickly this number explodes as the dataset size and the query length range increase:
considering the largest query length range ($S_{96-256}$) in the $20$GB dataset, we end up with a collection of subsequences (that need to be indexed) 5 orders of magnitude larger than the original dataset!
Computational time is significantly increased as well, since we have to construct different indexes for each query length we wish to support.

\begin{figure}[tb]
	\centering
	\includegraphics[trim={0cm 15cm 18cm 3cm},scale=0.65]{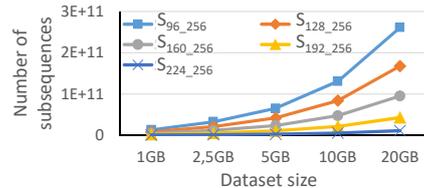}
	\caption{Search space evolution of variable length similarity search. Each dataset contains series of length $256$}
	\label{spaceExplosionVL}
\end{figure}

In the current literature, a technique based on multi-resolution indexes~\cite{914838,DBLP:journals/kais/KadiyalaS08} has been proposed in order to mitigate this explosion in size, by creating a smaller number of distinct indexes and performing more post-processing.
Nonetheless, this solution works exclusively for \emph{non} Z-normalized series\footnote{Z-normalization transforms a series so that it has a mean value of zero, and a standard deviation of one. This allows similarity search to be effective, irrespective of shifting (i.e., offset translation) and scaling\cite{DBLP:journals/datamine/KeoghK03}.} (which means that it cannot return results with similar trends, but different absolute values), and thus, renders the solution useless for a wide spectrum of applications.
Besides, it only mitigates the problem, since it still leads to a space explosion (albeit, at a lower rate), and therefore, it is not scalable, either.

We note that the technique discussed above (despite its limitations) is indeed the current state of the art, and no other technique has been proposed since, even though during the same period of time we have witnessed lots of activity and a steady stream of papers on the \emph{single-length} similarity search problem (e.g.,~\cite{DBLP:conf/vldb/PalpanasCGKZ04,Assent2008,shieh2008sax,DBLP:conf/icdm/CamerraPSK10,DBLP:journals/pvldb/WangWPWH13,DBLP:journals/kais/CamerraSPRK14,ZoumpatianosIP15,DBLP:journals/vldb/ZoumpatianosIP16,DBLP:conf/icdm/YagoubiAMP17,dpisaxjournal,DBLP:conf/bigdataconf/PengFP18,parisplus,peng2020messi,coconut,coconutpalm,DBLP:journals/vldb/KondylakisDZP19}). 
This attests to the challenging nature of the problem we are tackling in this paper.

\noindent{\bf Contributions.} 
In this work, we propose \textit{ULISSE} (ULtra compact Index for variable-length Similarity SEarch in data series), which is the first single-index solution that supports fast approximate and exact answering of variable-length (within a given range) similarity search queries for both non Z-normalized and Z-normalized data series collections. 
Our approach can be used with no changes with both Euclidean Distance and Dynamic Time Warping, for answering both \emph{k-NN} and $\epsilon$-range queries.

\textit{ULISSE} produces exact (i.e., correct) results, and is based on the following key idea: a data structure that indexes data series of length $\ell$, already contains all the information necessary for reasoning about any subsequence of length $\ell'<\ell$ of these series.
Therefore, the problem of enabling a data series index to answer queries of variable-length, becomes a problem of how to reorganize this information that already exists in the index.
To this effect, \textit{ULISSE} proposes a new summarization technique that is able to represent contiguous and overlapping subsequences, leading to succinct, yet powerful summaries: it combines the representation of several subsequences within a single summary, and enables fast (approximate and exact) similarity search for variable-length queries.

Our contributions can be summarized as follows\footnote{A preliminary version of this work has appeared elsewhere~\cite{ULISSEVldb,ulisse}.}: 
\begin{itemize}
	\item 
	We introduce the problem of Variable-Length Subsequences Indexing, which calls for a single index that can inherently answer queries of different lengths. 
	\item 
	We provide a new data series summarization technique, able to represent several contiguous series of different lengths. 
	This technique produces succinct, discretized envelopes for the summarized series, and can be applied to both non Z-normalized and Z-normalized data series. \item 
	Based on this summarization technique, we develop an indexing algorithm, which organizes the series and their discretized summaries in a hierarchical tree structure, namely, the \textit{ULISSE} index. 
	\item 
	We propose efficient exact and approximate k-NN algorithms, suitable for the \textit{ULISSE} index, which can compute similarity using either Euclidean Distance or Dynamic Time Warping measure. 
	\item 
	Finally, we perform an experimental evaluation with several synthetic and real datasets. 
	The results demonstrate the effectiveness and scalability of \textit{ULISSE} to dataset sizes that competing approaches cannot handle.
\end{itemize}

\noindent{\bf Paper Organization.}
The rest of this paper is organized as follows. 
Section~\ref{sec:relatedwork} discusses related work, and Section~\ref{sec:prelim} formulates the problem.
In Section~\ref{sec:ulisseFramework}, we describe the \textit{ULISSE} summarization techniques, and in Sections~\ref{sec:ulisseIndexing} and~\ref{sec:ulisseQuery} we explain our indexing and query answering algorithms.
Section~\ref{sec:experiments} describes the experimental evaluation, and we conclude in Section~\ref{sec:conclusions}.

\section{Related Work}
\label{sec:relatedwork}

\noindent{\bf Data series indexes.} 
The literature includes several techniques for data series indexing~\cite{DBLP:journals/pvldb/EchihabiZPB18},
which are all based on the same principle: they first reduce the dimensionality of the data series by applying some summarization technique (e.g., Piecewise Aggregate  Approximation (PAA)~\cite{Keogh2000}, or Symbolic Aggregate approXimation (SAX)~\cite{shieh2008sax}. 
However, all the approaches mentioned above share a common limitation: they can only answer queries of a fixed, predetermined length, which has to be decided before the index creation.

Faloutsos et al.~\cite{Faloutsos1994} proposed the first indexing technique suitable for variable length similarity search query. 
This technique extracts subsequences 
that are grouped in MBRs (Minimum Bounding Rectangles) and indexed using an R-tree.
We note that this approach works only for non Z-normalized sequences.
An improvement of this approach was proposed by Kahveci and Singh~\cite{914838}. 
They described MRI (Multi Resolution Index), which is a technique based on the construction of multiple indexes for variable length similarity search query. 

Storing subsequences at different resolutions (building indexes for different series lengths) provided a significant improvement over the earlier approach, since a greater part of a single query is considered during the search.
Subsequently, Kadiyala and Shiri~\cite{DBLP:journals/kais/KadiyalaS08} redesigned the MRI construction, in order to decrease the indexing size and construction time. 
This new indexing technique, called Compact Multi Resolution Index (CMRI), has a space requirement, which is 99\% smaller than the one of MRI. 
The authors also redefined the search algorithm, 
guaranteeing an improvement of the range search proposed upon the MRI index.

Loh et al.~\cite{DBLP:journals/datamine/LohKW04} proposed Index Interpolation for variable length subsequence similarity search for $\epsilon$-range queries.
This approach uses a single index that supports $\epsilon$-range search for subsequences of a fixed length that is smaller than the query length.
The search starts by considering a query prefix subsequence of the same length as the one supported by the index. 
During this process, the algorithm computes the distances between the original query and candidates of the same length, if the prefixes of these candidates have a distance to the query prefix smaller than the proposed bound.  
The authors proved the correctness of their solution, showing that their bounding strategy provides correct results for both non Z-normalized and Z-normalized subsequences.
We note that, based on the $\epsilon$-range search, this approach can also answer k-NN queries, thanks to the framework proposed by Han et al.~\cite{DBLP:conf/vldb/HanLMJ07}.

More recently, Wu et al.~\cite{KV-match} have proposed the KV-Match index, which supports $\epsilon$-range similarity search queries of variable length, using both Z-normalized Euclidean and DTW distances. 
The idea of this technique is similar to the CMRI one, since many indexes are built for different subsequence window lengths, which are considered at query time using multiple query segments. 
We note that for Z-normalized sequences, this method provides exact answers only for \textit{constrained} $\epsilon$-range search. 
To this effect, two new parameters that constrain the mean and the standard deviation of a valid result are considered at query answering time.

In contrast to CMRI and KV-Match, our approach uses a single index that is able to answer similarity search queries of variable length over larger datasets, and works for both non Z-normalized and Z-normalized series (a feature that is not supported by any of the previously introduced indexing techniques).


\noindent{\bf Sequential scan techniques.} 
Even though recent works have shown that sequential scans can be performed  efficiently~\cite{DBLP:conf/kdd/RakthanmanonCMBWZZK12,DBLP:conf/icdm/MueenHE14}, such techniques are mostly applicable when the dataset consists of a single, very long data series, and queries are looking for potential matches in small subsequences of this long data series. 
Such approaches, in general, do not provide any benefit when the dataset is composed of 
a large number of small data series, like in our case. 
Therefore, indexing is required in order to efficiently support data exploration tasks, which involve ad-hoc queries, i.e., the query workload is not known in advance.

\section{Preliminaries and Problem Formulation}
\label{sec:prelim}

Let a data series $D = d_1$,...,$d_{|D|}$ be a sequence of numbers $d_i \in  \mathbb{R}$, where $i \in \mathbb{N}$ represents the position in $D$. 
We denote the length, or size of the data series $D$ with $|D|$.
The subsequence $D_{o,\ell}$=$d_o$,...,$d_{o+\ell-1}$ of length $\ell$, is a contiguous subset of $\ell$ points of $D$ starting at offset $o$, where $1 \leq o \leq |D|$ and $ 1 \leq \ell \leq |D|-o+1$.
A subsequence is itself a data series.
A data series collection, $C$, is a set of data series. 

We say that a data series $D$ is Z-normalized, denoted $D^{n}$, when its mean $\mu$ is $0$ and its standard deviation $\sigma$ is $1$. 
The normalized version of $D = d_1,...,d_{|D|}$ is computed as follows: $D^{n} = \{\frac{d_1 - \mu}{\sigma},...,\frac{d_{|D|} - \mu}{\sigma}\}$. 
Z-normalization is an essential operation in several applications, because it allows similarity search irrespective of shifting and scaling~\cite{DBLP:journals/datamine/KeoghK03,DBLP:conf/kdd/RakthanmanonCMBWZZK12}.

\noindent{\bf Euclidean Distance.} Given two data series $D = d_{1},...,d_{|D|}$ and $D' = d'_{1},...,d'_{|D'|}$ of the same length (i.e., $|D| = |D'|$), we can calculate their Euclidean Distance 
as follows: $ED(D,D') = \sqrt{\sum_{i}^{|D|} d(d_{i},d'_{i})}$, where the distance function $d$ is applied to two real values, namely $A$ and $B$, as follows: $d(A,B) = (A-B)^2$.

\noindent{\bf Dynamic Time Warping.} The Euclidean distance is a lock-step measure, which is computed by summing up the distances between pairs of points that have the same positions in their respective series.
Dynamic Time Warping (DTW)~\cite{DTW} represents a more elastic measure, allowing for small mis-alignments of the matched points on the x-axis. 

Given two data series $D$ and $D'$, the DTW distance is computed by considering the differences between pairs of points $(d(d_{i},d'_{j}))$, where the indexes $i,j$ might be different.
In this manner, a particular alignment of $D$ and $D'$ is performed before to compute the distance. 
We define a sequence alignment as a vector of index pairs $A \in \mathbb{R}^{\ell \times 2}$, where $(i,j) \in A \iff 1 \le i,j \le \ell$, and $\ell$ is the length of the two series. 
The alignment of the Euclidean Distance is a special case, where the indexes are equal to their position in $A$.   
In the case of two series of length $\ell$, the space of the possible alignments spans the paths that join two cells in a squared matrix composed by $\ell^2$ cells. 
In Figure~\ref{matrixWarping}(a), we depict a Euclidean distance alignment of two series of length \textit{4}, which exactly crosses the diagonal of the matrix, joining the cells (\textit{1,1}) and (\textit{4,4}). 
On the other hand, in the same figure we report another possible alignment that we call warping alignment, which deviates from the diagonal. 
We use the terms \emph{warping path} and \emph{warping alignment} interchangeably.  

\begin{figure}[tb]
	\centering
	\includegraphics[trim={0cm 5cm 15cm 3cm},scale=0.5]{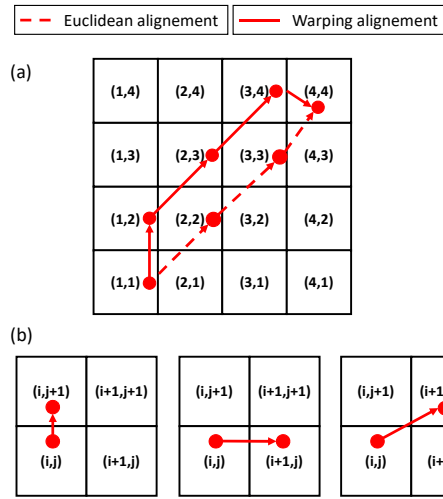}
	\caption{\textit{(a)} Euclidean and Warping alignment in a squared matrix. \textit{(b)} Valid index steps in a warping alignment.}
	\label{matrixWarping}
\end{figure}

In order to restrict the allowed paths, we can apply the following local constraints on the index pairs:
\begin{itemize}
	\item We require that the first and last pairs of $A$ correspond to the first and last pairs of points in $D$ and $D'$, respectively. If $|D|=|D'|=\ell$, we have $A[1]=(1,1)$ and $A[\ell]=(\ell,\ell)$.
	Furthermore, for any $(a,b),(c,d) \in A \iff (a \neq c) \vee (b \neq d)$. This latter, avoids to consider the same index pair twice in a single path. 
	\item Given $k \in \mathbb{N}$ ($1<k\le\ell$), we require that $A[k][0] - A[k-1][0] \le 1$ and $A[k][1] - A[k-1][1] \le 1$ always holds. This restricts each index to move by at most 1 unit to its next alignment position.
	\item Moreover, we always require that $A[k][0] - A[k-1][0] \ge 0$ and $A[k][1] - A[k-1][1] \ge 0$. This guarantees a monotonic movement of the path, towards the last index pair. In Figure~\ref{matrixWarping}(b), we depict the three possible steps that each index pair can perform in a valid alignment.     
\end{itemize} 

These constraints permit to bound the length of an alignment between two series of length $\ell$, between $\ell$ and $2 \times \ell-1$.    
Typically, warping paths are also subject to global constraints. We can thus set their maximum deviation from the matrix diagonal. In that regard, Sakoe and Chiba~\cite{Sakoe-Chiba} and Itakura~\cite{Itakura} proposed different warping path constraints, which restrict the matrix positions that a valid path can visit. 
The \textit{Sakoe-Chiba band}~\cite{Sakoe-Chiba} constraint allows each index of a warping path to be at most \textit{r} points far from the diagonal (Euclidean Distance alignment). On the other hand, the \textit{Itakura-parallelogram}~\cite{Itakura} constraint allows to choose different $r$ values depending on the index position $i$. In general, $r$ is called the \emph{warping window}.

Given a valid warping path, $A^{*}$, that satisfies the previously introduced constraints, we can formally define the DTW distance between two series $D$ and $D'$ of the same length $\ell$, as: $$DTW(D,D') = \underset{A^{*}}{\operatorname{argmin}}(\sqrt{\sum_{i}^{|A^{*}|} d(d_{A^{*}[i][0]},d'_{A^{*}[i][1]})}).$$
We note that computing the DTW distance corresponds to finding the valid alignment that minimizes the sum of the distances. 

In Figure~\ref{dtwMeasure}, we consider two series ($D$ and $D'$), which are extracted from two offsets that are \textit{5} points away, in the same long sequence. In this manner, the prefix of $D$ is equal to the suffix of $D'$, which starts at position \textit{6}. In the plots, the values of $D$ span the right vertical axis, whereas those of $D'$ the left one. If we compute the Euclidean distance, as depicted in Figure~\ref{dtwMeasure}(a), the fixed alignment of points does not capture the similarity of the two series. On the other hand, when computing the DTW distance, the warping path aligns the two similar parts, as reported in Figure~\ref{dtwMeasure}(b). At the bottom of the figure, we also report the warping path, which is constrained by a \textit{Sakoe-Chiba band}.

\begin{figure}[tb]
	\centering
	\includegraphics[trim={0cm 4cm 16cm 3cm},scale=0.56]{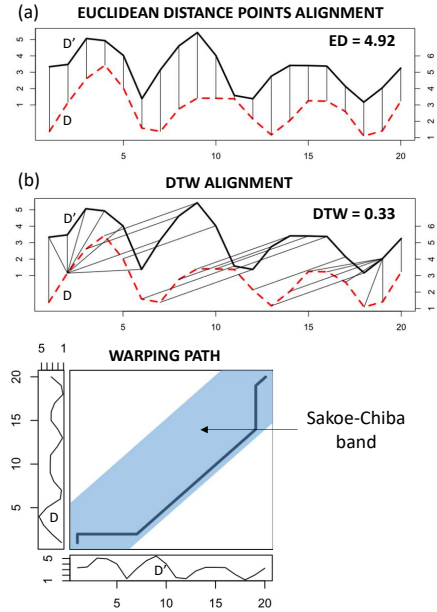}
	\caption{\textit{(a)} Euclidean distance alignment between the data series $D$ and $D'$. \textit{(b)} DTW Alignment between $D$ and $D'$. In the bottom part, the path is depicted in the $|D| \times |D'|$ matrix, contoured by the Sakoe-Chiba band.}
	\label{dtwMeasure}
\end{figure}

\noindent{\bf Problem Definition.} 
The problem we wish to solve in this paper is the following: 

\begin{problem}[Variable-Length Subsequences Indexing]
	Given a data series collection $C$, and a series length range $[\ell_{min},\ell_{max}]$, we want to build an index that supports exact similarity search, under the Euclidean and Dynamic Time Warping (DTW) measures, for queries of any length within the range $[\ell_{min},\ell_{max}]$. 
\end{problem}

In our case similarity search is formally defined as follows:

\begin{definition}[Similarity search]
	Given a data series collection $C=\{D^{1},...,D^{C}\}$, a series length range $[\ell_{min},\ell_{max}]$, a query data series $Q$, where $\ell_{min} \le |Q| \le \ell_{max}$, and $k \in \mathbb{N}$, we want to find the set $R=\{D^{i}_{o,\ell}| D^{i} \in C \wedge \ell=|Q| \wedge (\ell+o-1) \le |D^{i}|\}$, where $|R|=k$. We require that $ \forall D^{i}_{o,\ell} \in R$ $\nexists D^{i'}_{o',\ell'}$ $s.t.$ $dist(D^{i'}_{o',\ell'},Q) < dist(D^{i}_{o,\ell},Q)$, where $\ell'=|Q|$, $(\ell'+o'-1) \le |D^{i'}|$ and $D^{i'} \in C$. We informally call $R$, the $k$ \textit{nearest neighbors} set of $Q$. Given two generic series of the same length, namely $D$ and $D'$ the function $dist(d,d') = ED()$.
\end{definition}

In this work, we perform Similarity Search using either Euclidean Distance (ED) or Dynamic Time Warping (DTW), as the $dist$ function.



\subsection{The iSAX Index}


\begin{figure}[tb]
	\centering
	\hspace*{0.45cm}
	\includegraphics[trim={0cm 7cm 10cm 2cm},scale=0.58]{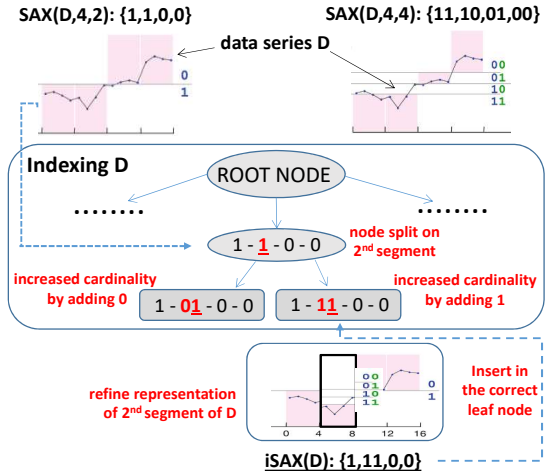}
	\caption{Indexing of series D (and an inner node split).}
	\label{FigurePAASAXISAX2}
\end{figure}

The Piecewise Aggregate Approximation (PAA) of a data series $D$, $PAA(D) = \{p_{1},...,p_{w}\}$, represents $D$ in a $w$-dimensional space by means of $w$ real-valued segments of length $s$, where the value of each segment is the mean of the corresponding values of $D$~\cite{Keogh2000}. 
We denote the first $k$ dimensions of $PAA(D)$, ($k \le w$), as  $PAA(D)_{1,..,k}$.
Then, the $iSAX$ representation of a data series $D$, denoted by $SAX(D,w,|alphabet|)$, is the representation of $PAA(D)$ by $w$ discrete coefficients, drawn from an alphabet of cardinality $|alphabet|$~\cite{shieh2008sax}. 

The main idea of the $iSAX$ representation (see Figure~\ref{FigurePAASAXISAX2}, top), is that the real-values space may be segmented by $|alphabet|-1$ breakpoints in $|alphabet|$ regions that are labeled by distinct symbols: binary values (e.g., with $|alphabet|=4$ the available labels are $\{00,01,10,11\}$). 
$iSAX$ assigns symbols to the $PAA$ coefficients, depending in which region they are located.    

The iSAX data series index is a tree data structure~\cite{shieh2008sax,DBLP:journals/kais/CamerraSPRK14}, consisting of three types of nodes (refer to Figure~\ref{FigurePAASAXISAX2}).
(i) The root node points to $n$ children nodes (in the worst case $n=2^w$, when the series in the collection cover all possible iSAX representations).
(ii) Each inner node contains the iSAX representation of all the series below it.
(iii) Each leaf node contains both the iSAX representation \emph{and} the raw data of all the series inside it (in order to be able to prune false positives and produce exact, correct answers). 
When the number of series in a leaf node becomes greater than the maximum leaf capacity, the leaf splits: it becomes an inner node and creates two new leaves, by increasing the cardinality of one of the segments of its iSAX representation. 
The two refined iSAX representations (new bit set to \textit{0} and \textit{1}) are assigned to the two new leaves. 

\section{The ULISSE framework}
\label{sec:ulisseFramework}


The key idea of the \textit{ULISSE} approach is the succinct summarization of \emph{sets} of series, namely, overlapping subsequences.
In this section, we present this summarization method.

\subsection{Representing Multiple Subsequences}

When we consider, contiguous and overlapping subsequences of different lengths within the range  $[\ell_{min},\ell_{max}]$ (Figure~\ref{ContainmentAreaEnvelope}(a), we expect the outcome as a bunch of similar series, whose differences are affected by the misalignment and  the different number of points. 
We conduct a simple experiment in Figure~\ref{ContainmentAreaEnvelope}(b), where we zero-align all the series shown in Figure~\ref{ContainmentAreaEnvelope}(a); we call those \textit{master series}. 

\begin{figure}[tb]
	\hspace*{-0.3cm}
	\includegraphics[trim={0cm 10cm 3cm 3cm},scale=0.68]{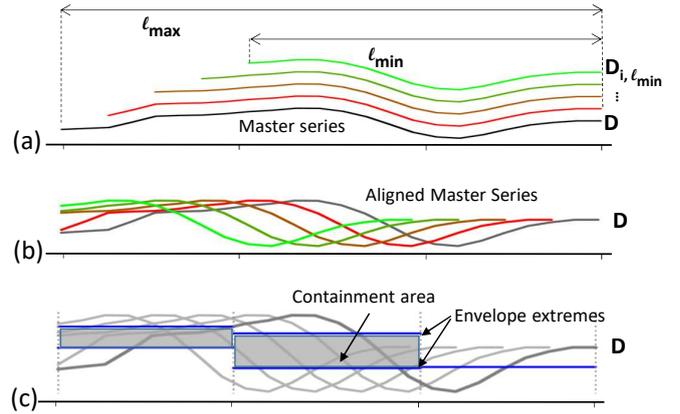}
	\caption{ \textit{a)} master series of $D$ in the length interval $\ell_{min},\ell_{max}$. \textit{b)} Zero-aligned master series. \textit{c)} Envelope built over the master series.}
	\label{ContainmentAreaEnvelope}
\end{figure}


\begin{definition}[Master Series]
	Given a data series $D$, and a subsequence length range $[\ell_{min},\ell_{max}]$, the master series are subsequences of the form $D_{i,min(|D|-i+1,\ell_{max})}$, for each $i$ such that $1 \le i \le |D|-(\ell_{min}-1)$, where $1 \le \ell_{min} \le \ell_{max} \le |D|$.
\end{definition}

We observe that the following property holds for the master series.
\begin{lemma}
	For any master series of the form $D_{i,\ell'}$, we have that $PAA(D_{i,\ell'})_{1,..,k} = PAA(D_{i,\ell''})_{1,..,k}$ holds for each $\ell''$ such that $\ell''\ge \ell_{min}$, $\ell''\le \ell' \le \ell_{max}$ and $\ell',\ell'' \% k = 0$.
\end{lemma}

\begin{proof}
	It trivially follows from the fact that, each non master series is always entirely overlapped by a master series. Since the subsequences are not subject to any scale normalization, their prefix coincides to the prefix of the equi-offset master series.
\end{proof}

Intuitively, the above lemma says that by computing only the $PAA$ of the master series in $D$, we are able to represent the $PAA$ prefix of any subsequence of $D$. 

When we zero-align the $PAA$ summaries of the master series, we compute the minimum and maximum $PAA$ values (over all the subsequences) for each segment: this forms what we call an \textit{Envelope} (refer to Figure~\ref{ContainmentAreaEnvelope}.c).
(When the length of a master series is not a multiple of the $PAA$ segment length, we compute the $PAA$ coefficients of the longest prefix, which is multiple of a segment.)
We call \textit{containment area} the space in between the segments that define the Envelope.

\subsection{PAA Envelope for Non-Z-Normalized Subsequences}

In this subsection, we formalize the concept of the \textit{Envelope}, introducing a new series representation.

We denote by $L$ and $U$ the $PAA$ coefficients, which delimit the lower and upper parts, respectively, of a containment area (see Figure~\ref{ContainmentAreaEnvelope}.c). 
Furthermore, we introduce a parameter $\gamma$, which corresponds to the number of master series we represent by the Envelope.
This allows to tune the number of subsequences of length in the range [$\ell_{min},\ell_{max}$], that a single Envelope represents, influencing both the tightness of a containment area and the size of the Index (number of computed Envelopes). We will show the effect of the relative tradeoff i.e., Tightness/Index size in the Experimental evaluation.
Given $a$, the point from where we start to consider the subsequences in $D$, and $s$, the chosen length of the PAA segment, we refer to an Envelope using the following signature:
\begin{align}
\qquad\qquad paaENV_{[D,\ell_{min},\ell_{max},a,\gamma,s]} = [L,U]
\end{align}


\subsection{PAA Envelope for Z-Normalized Subsequences}

So far we have considered that each subsequence in the input series $D$ is not subject of any scale normalization, i.e., is not Z-normalized. 
We introduce here a negative result, concerning the \emph{unsuitability} of a generic $paaENV_{[D,\ell_{min},\ell_{max},a,\gamma,s]}$ to describe subsequences that are Z-normalized.

Intuitively, we argue that the $PAA$ coefficients of a single master series $D_{i,a}$, generate a containment area, which may not embed the coefficients of the Z-normalized subsequence in the form $D'_{i,a'}$, for $a'<a$. 
This happens, because Z-normalization causes the subsequences of different lengths to change their shape, and even shift on the y-axis.
Figure~\ref{Z_norm_prob} depicts such an example. 

\begin{figure}[tb]
	\includegraphics[trim={0cm 14cm 0cm 3cm},scale=0.70]{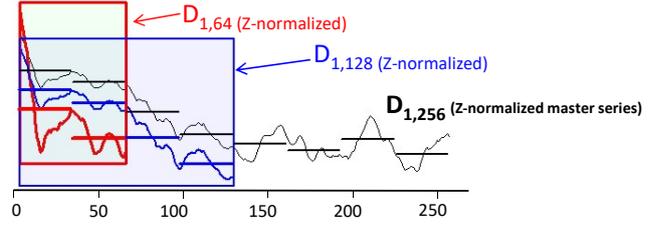}
	\caption{Master series $D_{1,256}$ with marked PAA coefficients.}
	\label{Z_norm_prob}
\end{figure}

We can now formalize this negative result.

\begin{lemma}
	A $paaENV_{[D,\ell_{min},\ell_{max},a,\gamma,s]}$ is not guaranteed to contain all the $PAA$ coefficients of the Z-normalized subsequences of lengths $[\ell_{min},\ell_{max}]$, of $D$. 
\end{lemma}

\begin{proof}
	To prove the correctness of the lemma, it suffices to pick such a case where a subsequence of $D$, namely $D_{a,\ell'}$, with $\ell_{min} \le \ell' \le \ell_{max}$, is not encoded by $paaENV_{[D,\ell_{min},\ell_{max},a,\gamma,s]}$. Formally, we should consider the case where $\exists k$ such that $PAA(D_{i,\ell'})_{k} > U_{k}$ or $PAA(D_{i,\ell'})_{k}  < L_{k}$. 
	We may pick a Z-normalized series $D$ choosing $\ell_{max} = |D| = \ell_{min}+1$ and $\gamma=0$. The resulting  $paaENV_{[D,\ell_{min}=\ell_{max}-1,\ell_{max}=|D|,i=1,\gamma=0,s]}$ obtains equal bounds, namely $L=U$. Let consider the z-normalized subsequence $D_{1,\ell_{min}}$. Its $PAA$ coefficients must be in the envelope. This implies that, $PAA(D_{1,\ell_{min}})_{1}= L_1 = U_1 \label{cvequation}$ must hold.
	If $s$ is the $PAA$ segment length, in the case of Z-normalization, $PAA(D_{1,\ell_{min}})_{1} = (((\sum_{i=1}^{s} d_i) - (\mu_{D_{1,\ell{min}}} \times s)) / \sigma_{D_{1,\ell{min}}})/s$ and $U_1 = (((\sum_{i=1}^{s} d_i) - (\mu_{D} \times s)) / \sigma_{D})/s$. Therefore, the following equation: $ (\mu_{D_{1,\ell{min}}} \times s) / \sigma_{D_{1,\ell{min}}} = (\mu_{D} \times s) / \sigma_{D}$ holds, which is equivalent to $\mu_{D_{1,\ell{min}}} / \sigma_{D_{1,\ell{min}}} = \mu_{D} / \sigma_{D}$. 
	At this point we may have that $\mu_{D} = \mu_{D_{1,\ell{min}}} $, when $D_{\ell{max},1} = \mu_{D_{1,\ell{min}}}$. This clearly leads to have a smaller dispersion on $D$ than $D_{1,\ell{min}}$ and thus $\sigma_{D} < \sigma_{D_{1,\ell{min}}} \Longrightarrow PAA(D_{1,\ell_{min}})_{1} \ne L_1 \ne U_1$. 
\end{proof}



If we want to build an Envelope, containing all the Z-normalized sequences, we need to take into account the shifted coefficients of the Z-normalized subsequences, which are not master series. 
Hence, each $PAA$ segment coefficient (in a master series) will be represented by the set of values resulting from the Z-normalizations of all the subsequences of length in [$\ell_{min},\ell_{max}$] that are not master series and contain that segment.

Given a generic master series $D_{i,\ell} =\{d_i,..d_{i+\ell-1}\}$, and $s$ the length of the segment, its $k^{th}$ $PAA$ coefficient set is computed by:
\inlineequation[setPaaNorm]{PAA^{*}(D_{i,\ell})_{k} = \{( \frac{(\sum_{p=s(k-1)+1}^{s(k-1)+s} d_p) - (\mu_{D_{i,\ell'}} \times s)}{  \sigma_{D_{i,\ell'}}   }) / s | \ell_{min} \le \ell' \le \ell_{max}, \ell' \ge (s(k-1)+s-(i-1))   \}}.

In Figure~\ref{exPaaStar}, we depict an example of $PAA^{*}$ computation for the first segment of the master series $D$.

\begin{figure}[tb]
	\includegraphics[trim={0cm 15cm 0cm 3cm},scale=0.72]{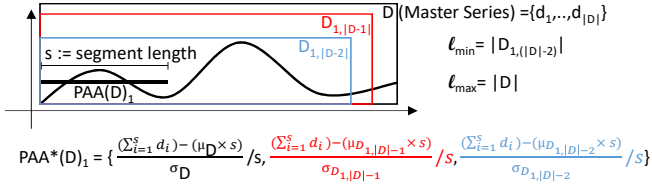}
	\caption{$PAA^{*}(D)_{1}$ computation. Since the first PAA segment (of length $s$) of the master series $D$, is also the first one of the two non master series $D_{1,|D-1|}$, $D_{1,|D-2|}$, three PAA coefficients are computed with the different normalizations.} 
	\label{exPaaStar}
\end{figure}

We can then follow the same procedure as before (in the case of non Z-normalized sequences), computing the minimum and maximum $PAA$ coefficients for each segment given by the above formula, in order to get the Envelope for the Z-normalized sequences (which we also denote with $paaENV$).

\subsection{Indexing the Envelopes}

Here, we define the procedure used to index the Envelopes. 
In that regard, we aim to adapt the $iSAX$ indexing mechanism (depicted in Figure~\ref{FigurePAASAXISAX2}).

Given a $paaENV$, we can translate its $PAA$ extremes into the relative iSAX representation: $uENV_{ paaENV_{[D,\ell_{min},\ell_{max},a,\gamma,s]}}=[iSAX(L),iSAX(U)]$,
where $iSAX(L)$ ($iSAX(U)$) is the vector of the minimum (maximum) $PAA$ coefficients of all the segments corresponding to the subsequences of $D$.

The \textit{ULISSE} Envelope, $uENV$, represents the principal building block of the \textit{ULISSE} index.
Note that, we might remove for brevity the subscript containing the parameters from the $uENV$ notation, when they are explicit.

\begin{figure}[tb]
	\hspace*{-0.4cm}
	\includegraphics[trim={-0.5cm 10cm 0cm 3cm},scale=0.8]{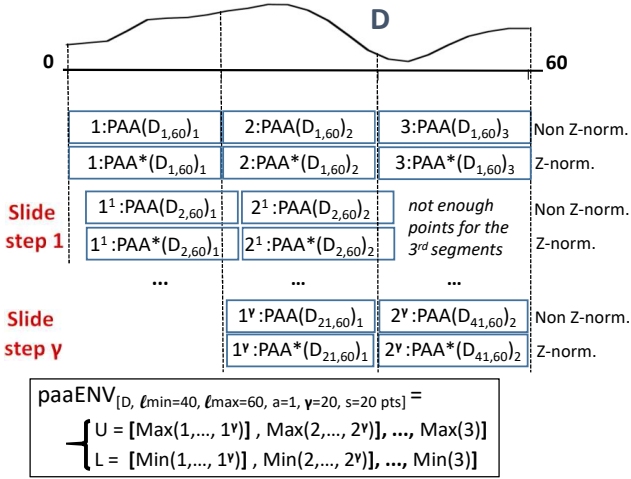}
	\caption{ $uENV$ building, with input: data series $D$ of length $60$, $PAA$ segment size $=20$, $\gamma = 20$, $\ell_{min} = 40$ and $\ell_{max} = 60$.}
	\label{BuildEnvelope}
\end{figure}

In Figure~\ref{BuildEnvelope}, we show a small example of envelope building, given an input series $D$. 
The picture shows the $PAA$ coefficients computation of the master series. 
They are calculated by using a sliding window starting at point $a=1$, which stops after $\gamma$ steps.
Note that the Envelope generates a containment area, which embeds all the subsequences of $D$ of all lengths in the range $[\ell_{min},\ell_{max}]$.

\section{Indexing Algorithm}
\label{sec:ulisseIndexing}

\subsection{Indexing Non-Z-Normalized Subsequences}
We are now ready to introduce the algorithms for building an $uENV$. 
Algorithm~\ref{AlgoEnvelopeNN} describes the procedure for non-Z-normalized subsequences.  
As we noticed, maintaining the running sum of the last $s$ points, i.e., the length of a $PAA$ segment (refer to Line~\ref{runSum}), allows us to compute all the $PAA$ values of the expected envelope in $O(w(\ell_{max}+\gamma))$ time in the worst case, where $\ell_{max}+\gamma$ is the points window we need to take into account for processing each master series, and $w$ is the number of $PAA$ segments in the maximum subsequence length $\ell_{max}$. 
Since $w$, is usually a very small number (ranging between 8-16), it essentially plays the role of a constant factor. 
In order to consider not more than $\gamma$ steps for each segment position, we store how many times we use it, to update the final envelope in the vector, in Line~\ref{updatesSegment}.  

\begin{algorithm}[tb]
	{	\scriptsize
		\SetAlgoLined
		\label{AlgoEnvelopeNN}
		\KwIn{\textbf{float}[] $D$, \textbf{int} $s$, \textbf{int} $\ell_{min}$, \textbf{int} $\ell_{max}$, \textbf{int} $\gamma$, \textbf{int} $a$ }
		\KwOut{$\mathbf{uENV[iSAX_{min},iSAX_{max}]}$}
		\scriptsize
		\BlankLine
		\textbf{int} w $\leftarrow$ $\lfloor\ell_{max} / s\rfloor$ \;
		\textbf{int} segUpdateList[S] $\leftarrow$ \{0,...,0\}\;\label{updatesSegment}
		\textbf{float} $U[w]\leftarrow\{-\infty,...,-\infty\}$, $L[w]\leftarrow$ $\{\infty,...,\infty\}$\;
		\uIf{$|D|-(i-1) \geq \ell_{min}$ } 
		{
			\textbf{float} paaRSum $\leftarrow$ 0\;
			\tcp{iterate the master series.}
			\For{i $\leftarrow$ a \emph{\KwTo} min($|D|$,$a+\ell_{max}+\gamma$) } 
			{	\tcp{running sum of paa segment}
				paaRSum $\leftarrow$ paaRSum + D[i]\; 	\label{runSum}
				\If{(j-a) $>$ $s$}
				{
					paaRSum $\leftarrow$ paaRSum - D[i-s]\;\label{runSumUpdate}
				}
				\For{z $\leftarrow$ 1 \emph{\KwTo} min($\lfloor$[i-(a-1)] $/$ s$\rfloor$,w)}
				{
					\If{segUpdatedList[z] $\le$ $\gamma$}
					{
						segUpdateList[z] ++\;
						\textbf{float} paa $\leftarrow$ (paaRSum $/$ s)\;
						$L[z]$ $\leftarrow$ $min(paa,L[z])$\;    
						$U[z]$ $\leftarrow$ $max(paa,U[z])$\;
					} 	
				}
			}
			$\mathbf{uENV}$ $\leftarrow$ [iSAX($L$),iSAX($U$)];\                                  
		}
		\Else{
			$\mathbf{uENV}$ $\leftarrow \emptyset$\;
		}
	}
	\caption{$uENV computation$}
\end{algorithm}

\subsection{Indexing Z-Normalized Subsequences}

In Algorithm~\ref{algoEnvNorm}, we show the procedure that computes an indexable Envelope for Z-normalized sequences, which we denote as $uENV_{norm}$. 
This routine iterates over the points of the overlapping subsequences of variable length (\textit{First loop} in Line~\ref{loopMasterSeries}), and performs the computation in two parts. 
The first operation consists of computing the sum of each $PAA$ segment we keep in the vector $PAAs$ defined in Line~\ref{vectorSegmentSum}. When we encounter a new point, we update the sum of all the segments that contain that point (Lines~\ref{startUpdSeg}-\ref{endUpdSeg}).
The second part, starting in Line~\ref{loopNonMasterandMasterSeries} (\textit{Second loop}), performs the segment normalizations, which depend on the statistics (mean and std.deviation) of all the subsequences of different length (master and non-master series), in which they appear.
During this step, we keep the sum and the squared sum of the window, which permits us to compute the mean and the standard deviation in constant time (Lines~\ref{mean},\ref{gamma}). We then compute the Z-normalizations of all the $PAA$ coefficients in Line~\ref{Z_normCoefficient}, by using Equation~\ref{setPaaNorm}.  

\begin{algorithm}[!tb]
	{\scriptsize
		\SetAlgoLined
		\label{algoEnvNorm}
		\KwIn{\textbf{float}[] $D$, \textbf{int} $s$, \textbf{int} $\ell_{min}$, \textbf{int} $\ell_{max}$, \textbf{int} $\gamma$, \textbf{int} a }
		\KwOut{$\mathbf{uENV_{norm}[iSAX_{min},iSAX_{max}]}$}
		\BlankLine
		
		\textbf{int} w $\leftarrow$ $\lfloor\ell_{max} / s\rfloor$ \; \label{segmentsNumber}
		\tcp{sum of PAA segments values }
		\textbf{float} PAAs [$\ell_{max} + \gamma - (s-1)$] $\leftarrow$ \{0,...,0\}\; 
		\label{vectorSegmentSum}
		\textbf{float} $U[w]\leftarrow\{-\infty,...,-\infty\}$, $L[w]\leftarrow \{\infty,...,\infty\}$\;
		\uIf{$|D|-(a-1)$ $\geq$ $\ell_{min}$} 
		{
			\textbf{int} nSeg$\leftarrow$ 1\;
			\textbf{float} accSum, accSqSum $\leftarrow$ 0\;
			\tcp{First loop: Iterate the points.}
			\For{i $\leftarrow$ a \emph{\KwTo} min($|D|$,(a+$\ell_{max}$+$\gamma$))} 
			{\label{loopMasterSeries}
				\tcp{update sum of PAA segments values}
				\If{$i-a>$ s}
				{\label{startUpdSeg}
					nSeg++\;
					PAAs[nSeg] $\leftarrow$ PAAs[nSeg-1] - D[i-s]\; 
				}
				PAAs[nSeg] += D[i]\; \label{endUpdSeg}
				\tcp{keep sum and squared sum.}
				accSum 	+= D[i], accSqSum += (D[i])$^2$\;\label{sumSquaredSum}	
				\tcp{the window contains enough points.}		
				\If{i-(a-1) $\geq \ell_{min}$}
				{
					acSAc $\leftarrow$ accSum, 
					acSqSAc $\leftarrow$ accSqSum\;
					\textbf{int} nMse $\leftarrow$ min($\gamma$+1,(i-(a-1)-$\ell_{min})+1$)\;
					\tcp{Normalizations of PAA coefficients.}
					\For{j $\leftarrow$ 1 \emph{\KwTo} nMse }
					{\label{loopNonMasterandMasterSeries}
						\textbf{int} wSubSeq $\leftarrow$ i-(a-1)-(j-1) \;			
						\If{wSubSeq $\leq \ell_{max}$}
						{
							\textbf{float} $\mu$ $\leftarrow$acSAc/wSubSeq\; \label{mean}
							\textbf{float} $\sigma$ $\leftarrow$ \label{gamma} $\sqrt{(\frac{acSqSAc}{wSubSeq} -\mu^2)}$\;					
							\textbf{int} nSeg $\leftarrow \lfloor$wSubSeq$\div s \rfloor $\;
							\For{z $\leftarrow$ 1 \emph{\KwTo} nSeg }
							{
								\textbf{float} a $\leftarrow$ PAAs[j+[(z-1)$\times$s]]\; 
								\textbf{float} b $\leftarrow$ s$\times\mu$\;  
								\textbf{float} paaNorm $\leftarrow$ $\frac{((a-b) / \sigma)}{s}$\;	\label{Z_normCoefficient}
								$L[z]$ $\leftarrow$ $min(paaNorm,L[z])$\;    
								$U[z]$ $\leftarrow$ $max(paaNorm,U[z])$\;
								\label{endComputingNormalization}			
							}
						}
						acSAc -= D[j], acSqSAc -= (D[j])$^2$\;
					}
				}	
			}
			
			$\mathbf{uENV_{norm}}$ $\leftarrow$ [iSAX($L$),iSAX($U$)];\                                  
		}
		\Else{
			$\mathbf{uENV_{norm}}$ $\leftarrow \emptyset$\;
		}
	}
	
	\caption{$uENV_{norm}$ computation}
	
\end{algorithm}

\begin{figure}[tb]
	\includegraphics[trim={0cm 9.5cm 0cm 3cm},scale=0.80]{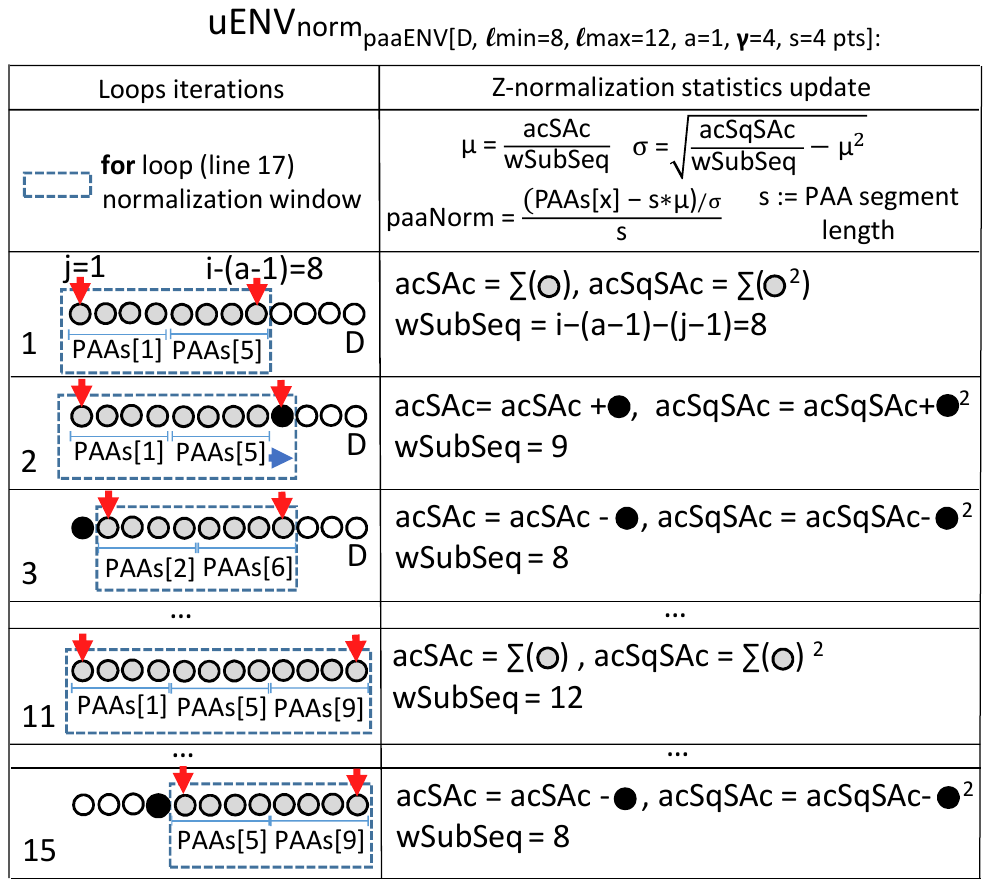}
	\caption{Running example of Algorithm~\ref{algoEnvNorm}. \textit{Left column}) Points iteration, the dashed squared contours the subsequence used to normalize the PAA coefficients in the Second loop. \textit{Right column}) Statistics update at each step, which serve the computation of $\mu$ and $\sigma$ of each possible coefficients normalization.}
	\label{exAlgo2}
\end{figure}

In Figure~\ref{exAlgo2}, we show an example that illustrates the operation of the algorithm.
In \textit{1}, the \textit{First loop} has iterated over \textit{8} points (marked with the dashed square). Since they form a subsequence of length $\ell_{min}$, the \textit{Second loop} starts to compute the Z-normalized PAA coefficients of the two segments, computing the mean and the standard deviation using the sum ($acSAc$) and squared sum ($acSqAc$) of the points considered by the \textit{First loop} (gray circles).
The second step takes place after that the \textit{First loop} has considered the $9^{th}$ point (black circle) of the series. Here, the \textit{Second loop} updates the sum and the squared sum, with the new point, calculating then the corresponding new Z-normalized PAA coefficients. At step \textit{3}, the algorithm considers the second subsequence of length $\ell_{min}$, which is contained in the nine points window.
The \textit{Second loop} considers in order all the overlapping subsequences, with different prefixes and length.
This permits to update the statistics (and all possible normalizations) in constant time.
The algorithm terminates, when all the points are considered by the \textit{First loop}, and the \textit{Second loop} either encounters a subsequence of length $\ell{min}$ (as depicted in the step \textit{15}), or performs at most $\gamma$ iterations, since all the subsequences starting at position $a+\gamma+1$ or later (if any) will be represented by other Envelopes.   

\subsubsection{Complexity Analysis} \textit{}\\
Given $w$, the number of PAA segments in the window of length $\ell_{max}$, and $M=\ell_{max}-\ell_{min}+\gamma$, the number of master series we need to consider, building a normalized Envelope, $uENV_{norm}$, takes $O(M \gamma w$) time.\\ 



\subsection{Building the index}

\begin{algorithm}[!tb]
	{\scriptsize
		\SetAlgoLined
		\label{algoBuildUlisse}
		\KwIn{\textbf{Collection} $C$, \textbf{int} $s$, \textbf{int} $\ell_{min}$, \textbf{int} $\ell_{max}$, \textbf{int} $\gamma$, \textbf{bool} $bNorm$ }
		\KwOut{\textbf{\textit{ULISSE} index} I}
		\scriptsize
		\ForEach{D \textbf{in} C}
		{
			$\mathbf{int} a' \leftarrow \emptyset$\;
			$\mathbf{uENV}$ $E \leftarrow \emptyset$\;
			\While{true}
			{	
				\uIf{bNorm}
				{
					$E \leftarrow uENV_{norm}(D,s,\ell_{min},\ell_{max},\gamma,a')$\;
				}
				\Else
				{
					$E \leftarrow uENV(D,s,\ell_{min},\ell_{max},\gamma,a')$\;
				}	
				$a' \leftarrow a' + \gamma + 1$ \;
				\If{$E == \emptyset$}
				{
					\textbf{break}\;
				}
				$bulkLoadingIndexing(I,E)$\;\label{bulkloading}
				$I.inMemoryList.add(maxCardinality(E))$\;\label{listEvelope}
				
			}
		}  
	}
	\caption{\textit{ULISSE} index computation}	
\end{algorithm}

We now introduce the algorithm, which builds a \textit{ULISSE} index upon a data series collection. 
We maintain the structure of the $iSAX$ index~\cite{DBLP:journals/kais/CamerraSPRK14}, introduced in the preliminaries. 

Each \textit{ULISSE} internal node stores the Envelope $uENV$ that represents all the sequences in the subtree rooted at that node. 
Leaf nodes contain several Envelopes, which by construction have the same $iSAX(L)$. 
On the contrary, their $iSAX(U)$ varies, since it get updated with every new insertion in the node. 
(Note that, inserting by keeping the same $iSAX(U)$ and updating $iSAX(L)$ represents a symmetric and equivalent choice.)

In Figure~\ref{ULiSSE_index}, we show the structure of the \textit{ULISSE} index during the insertion of an Envelope (rectangular/yellow box). 
Note that insertions are performed based on $iSAX(L)$ (underlined in the figure). 
Once we find a node with the same $iSAX(L)=(1-0-0-0)$ (Figure~\ref{ULiSSE_index}, $1^{st} step$)
if this is an inner node, we descend its subtree (always following the $iSAX(L)$ representations) until we encounter a leaf.
During this path traversal, we also update the $iSAX$ representation of the Envelope we are inserting, by increasing the number of bits of the segments, as necessary. 
In our example, when the Envelope arrives at the leaf, it has increased the cardinality of the second segment to two bits: $iSAX(L)=(1-${\bf 10}$-0-0)$, and similarly for $iSAX(U)$ (Figure~\ref{ULiSSE_index}, $2^{nd} step$).
Along with the Envelope, we store in the leaf a pointer to the location on disk for the corresponding raw data series.
We note that, during this operation, we do not move any raw data into the index.

To conclude the insertion operation, we also update the $iSAX(U)$ of the nodes visited along the path to the leaf, where the insertion took place.
In our example, we update the upper part of the leaf Envelope to $iSAX(U)=(1-${\bf 11}$-0-0)$, as well as the upper part of the Envelope of the leaf's parent to $iSAX(U)=(1-${\bf 1}$-0-0)$ (Figure~\ref{ULiSSE_index}, $3^{rd} step$).
This brings the \emph{ULISSE} index to a consistent state after the insertion of the Envelope.

Algorithm~\ref{algoBuildUlisse} describes the procedure, which iterates over the series of the input collection $C$, and inserts them in the index. 
Note that function $bulkLoadingIndexing$ in Line~\ref{bulkloading} may use different bulk loading techniques. In our experiments, we used the iSAX 2.0 bulk loading algorithm \cite{DBLP:conf/icdm/CamerraPSK10}.
Alongside the index, we also keep in memory (using the raw data order) all the Envelopes, represented by the symbols of the highest $iSAX$ cardinality available (Line~\ref{listEvelope}). 
This information is used during query answering.

\begin{figure}[tb]
	\includegraphics[trim={0cm 8cm 0cm 3cm},scale=0.63]{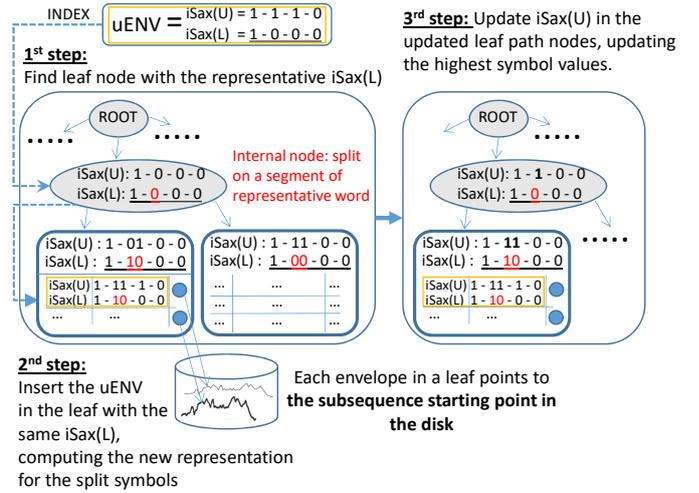}
	\caption{Envelope insertion in an \textit{ULISSE} index. $iSAX(L)$ is chosen to accommodate the Envelopes inside the nodes.}
	\label{ULiSSE_index}
\end{figure}

\subsubsection{Space complexity analysis}
The index space complexity is equivalent for the case of Z-normalized and non Z-normalized sequences. The choice of $\gamma$ determines the number of $Envelopes$ generated and thus the index size.
Hence, given a data series collection $C=\{D^{1},...,D^{|C|}\}$ the number of extracted Envelopes is given by $N= (\sum_{i}^{|C|} \lfloor \frac{|D^{i}|}{\ell_{min} + \gamma}\rfloor$). 
If $w$ PAA segments are used to discretize the series, each $iSAX$ symbol is represented by a single byte (binary label) and the disk pointer in each Envelope occupies $b$ bytes (in general \textit{8} bytes are used). The final space complexity is $O((2w)bN)$.

\section{Similarity Search with ULISSE}
\label{sec:ulisseQuery}

In this section, we present the building blocks of the similarity search algorithms we developed for the \textit{ULISSE} index, for both the Euclidean and the DTW distances, and both \emph{k-NN} and $\epsilon$-range queries. 

We note that the same index structure supports both distance measures. 
When the query arrives, and depending on the distance measure we have chosen, we use the corresponding lower bounding and real distance formulas.
We elaborate on these procedures in the following sections.

\subsection{Lower Bounding Euclidean Distance}

The iSAX representation allows the definition of a distance function, which lower bounds the true Euclidean~\cite{shieh2008sax}. 
This function compares the $PAA$ coefficients of the first data series, against the iSAX breakpoints (values) that delimit the symbol regions of the second data series. 

Let $\beta_{u}(S)$ and $\beta_{l}(S)$ be the breakpoints of the iSAX symbol $S$.
We can compute the distance between a $PAA$ coefficient and an iSAX region using:

\begin{equation}
\begin{gathered}
{\small distLB(PAA(D)_{i},iSAX(D')_{i}) = } \\
\begin{cases}
\scriptstyle (\beta_{u}(iSAX(D')_{i}) - PAA(D)_{i})^2&if \scriptstyle \beta_{u} (iSAX(D')_{i}) < PAA(D)_{i} \\
\scriptstyle (\beta_{l}(iSAX(D')_{i}) - PAA(D)_{i})^2&if \scriptstyle \beta_{l}(iSAX(D')_{i})  > PAA(D)_{i} \\
\scriptstyle 0 & {\footnotesize \text{otherwise.}}
\end{cases}	
\end{gathered}
\end{equation}

In turn, the lower bounding distance between two equi-length series $D$,$D'$, represented by $w$ PAA segments and $w$ $iSAX$ symbols, respectively, is defined as:

\begin{equation}
\small
\label{mindistPAAiSAX}
\begin{gathered}
\qquad mindist_{PAA\_iSAX}(PAA(D),iSAX(D')) = \\   \sqrt{\frac{|D|}{w}} \sqrt{\sum_{i=1}^{w} distLB(PAA(D)_{i},iSAX(D')_{i})}.
\end{gathered}
\end{equation}

We rely on the following proposition~\cite{Lin2007}:

\begin{proposition}
	Given two data series $D,D'$, where {\small $|D|=|D'|$,	 $mindist_{PAA\_iSAX}(PAA(D),iSAX(D')) \le ED(D,D')$}. 
\end{proposition}

Since our index contains Envelope representations, we need to adapt Equation~\ref{mindistPAAiSAX}, in order to lower bound the distances between a data series $Q$, which we call query, and a set of subsequences, whose iSAX symbols are described by the Envelope\\ $uENV_{paaENV_{[D,\ell_{min},\ell_{max},a,\gamma,s]}}=[iSAX(L),iSAX(U)]$. 

Therefore, given $w$, the number of PAA coefficients of $Q$, that are computed using the Envelope PAA segment length $s$ on the longest multiple prefix, we define the following function:

\begin{equation}
\begin{gathered} 
\label{mindistPAAuENV}
\qquad {\small mindist_{ULiSSE}(PAA(Q),uENV_{paaENV_{...}})=}\\
\small \sqrt{s} \small \sqrt{\sum_{i=1}^{w} \begin{cases}
	\scriptstyle (PAA(Q)_{i}-\beta_{u}(iSAX(U)_{i}))^2,& \scriptstyle if (*)\\
	\scriptstyle (PAA(Q)_{i}-\beta_{u}(iSAX(L)_{i}))^2,& \scriptstyle if (**)\\
	\scriptstyle 0 &  {\footnotesize \text{otherwise.}}
	\end{cases}}\\
(*) \scriptstyle \beta_{u} (iSAX(U)_{i}) < PAA(Q)_{i}\\
(**) \scriptstyle \beta_{l} (iSAX(L)_{i}) > PAA(Q)_{i}
\end{gathered} 
\end{equation}

\begin{figure}[tb]
	\centering
	\includegraphics[trim={0cm 15.5cm 20cm 3cm},scale=1.1]{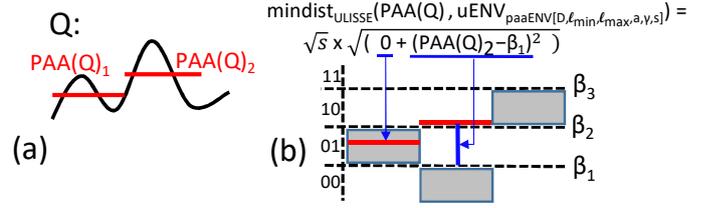}
	\caption{Given the PAA representation of a query $Q$ (\textit{a}) and $uENV_{ paaENV_{[D,\ell_{min},\ell_{max},a,\gamma,s]}}$ (\textit{b}) we compute their $mindist_{ULiSSE}$. The $iSAX$ space is delimited with dashed lines and the relative breakpoints $\beta_{i}$.}
	\label{ExLBcomp}
\end{figure}

In Figure~\ref{ExLBcomp}, we report an example of $mindist_{ULiSSE}$ computation between a query $Q$, represented by its PAA coefficients, and an Envelope in the iSAX space.

\begin{proposition}
	Given two data series
	{\small  $Q$,$D$,\\$mindist_{ULiSSE}(PAA(Q),uENV_{paaENV_{[D,\ell_{min},\ell_{max},a,\gamma,s]}})\le ED(Q,D_{i,|Q|})$, for each $i$ such that $a \le i \le a+\gamma+1$ and $|D|-(i-1) \ge \ell_{min}$
	}. 
\end{proposition}

\begin{proof}(sketch)
	We may have two cases, when $mindist_{ULiSSE}$ is equal to zero, the proposition clearly holds, since Euclidean distance is non negative. 
	On the other hand, the function yields values greater than zero, if one of the first two branches is true. 
	Let consider the first (the second is symmetric). 
	If we denote with $D''$ the subsequence in $D$, such that $ \beta_{l} (iSAX(U)_{i}) \le PAA(D'')_{i} \le \beta_{u} (iSAX(U)_{i})$, we know that the upper breakpoint of the $i^{th}$ iSAX symbol, of each subsequence in $D$, which is represented by the Envelope, must be less or equal than $\beta_{u}(iSAX(U)_{i})$. 
	It follows that, for this case, Equation~\ref{mindistPAAuENV} is equivalent to\\$distLB(PAA(Q)_{i},iSAX(D'')_{i})$, which yields the shortest lower bounding distance between the $i^{th}$ segment of points in $D$ and $Q$. 
\end{proof}

\subsection{Lower Bounding Dynamic Time Warping}

\begin{figure}[tb]
	\centering
	\hspace*{-0.3cm}
	\includegraphics[trim={0cm 4cm 18cm 3cm},scale=0.78]{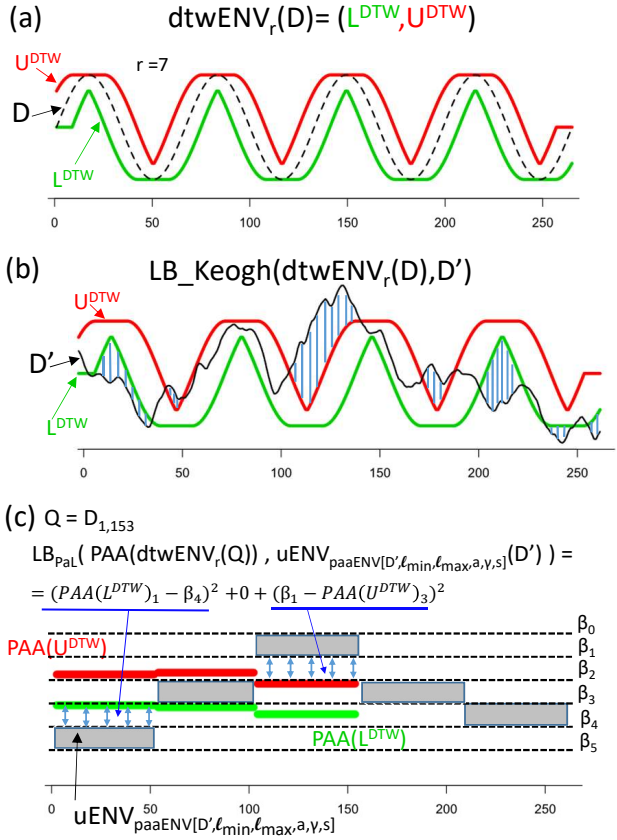}
	\caption{\textit{(a)} DTW Envelope ($L^{DTW}$, $U^{DTW}$) of a series $D$. \textit{(b)} $LB_{Keogh}$ distance between DTW Envelope and $D'$. \textit{(c)} $LB_{PaL}$ between the DTW Envelope of $Q$ (prefix of $D$) and the \textit{ULISSE} Envelope of $D'$. The horizontal dashed lines delimit the iSAX breakpoints space.}
	\label{lowerboundingDTW}
\end{figure}

We present here a lower bound for the true DTW distance between two data series. 
Keogh et al.~\cite{DBLP:journals/kais/KeoghR05} introduced the $LB_{Keogh}$ function, which provides a measure that is always smaller or equal than the true DTW, between two equi-length series.
To compute this measure, we need to account for the valid warping alignments of two data series. 
Recall that the indexes of a valid path are confined by the Sakoa-Chiba band, where they are at most $r$ points far from the diagonal (Euclidean Distance alignment). 
Given a data series $D$, we can build an envelope, $dtwENV_{r}(D)$, composed by two data series: $L^{DTW}$ and $U^{UDTW}$, which delimit the space generated by the points of $D$ that have indexes in the valid warping paths, constrained by the window $r$.
Therefore, the $i^{th}$ point of the two envelope sequences are computed as follows: $L_{i}^{DTW} = min(D_{(i-r,2r+1)})$ and $U_{i}^{DTW} = max(D_{(i-r,2r+1)})$. 
Intuitively, each $i^{th}$ value of $L^{DTW}$ and $U^{DTW}$ represent the minimum and the maximum values, respectively, of the points in $D$ that can be aligned with the $i^{th}$ position of a matching series. 
In Figure~\ref{lowerboundingDTW}(a), we report a data series $D$ (plotted using a dashed line), contoured by its $dtwENV_{r}(D)$ envelope ($r=7$).

\noindent{\bf Lower bounding DTW.} We can thus define the $LB_{Keogh}$ distance~\cite{DBLP:journals/kais/KeoghR05}, which is computed between a DTW envelope of a series $D$ and a data series $D'$, where $|D|=|D'|$ and the warping window is $r$:

{\small
	\begin{equation}
	\begin{gathered} 
	\qquad\qquad LB_{Keogh}(dtwENV_{r}(D),D') =\\
	\qquad\qquad \sqrt{\sum_{i=1}^{|D|} 
		\begin{cases}
		( D'_{i}-U_{i}^{DTW})^2,&  if  D'_{i} > U_{i}^{DTW}\\
		( D'_{i}-L_{i}^{DTW})^2,&  if  D'_{i} < L_{i}^{DTW}\\s
		0 &  \text{otherwise.}
		\end{cases}}
	\label{lbKeogh}
	\end{gathered}
	\end{equation}
}

The $LB_{Keogh}$ distance between $dtwENV_{r}(D)$ and $D'$ is guaranteed to be always smaller than, or equal to $DTW(D,D')$, computed with warping window $r$.
In Figure~\ref{lowerboundingDTW}(b), we depict the $LB_{Keogh}$ distance between $dtwENV_{r}(D)$ (from Figure~\ref{lowerboundingDTW}(a)), and a new series $D'$. 
The vertical (blue) lines represent the positive differences between $D'$ and the DTW envelope of $D$, in Equation~\ref{lbKeogh}.
Note that the computation of $LB_{Keogh}$ takes $O(\ell)$ time (linear), whereas the true DTW computation runs in $O(\ell r)$ time using dynamic programming~\cite{DBLP:journals/kais/KeoghR05,DBLP:conf/kdd/RakthanmanonCMBWZZK12}.

\noindent{\bf Lower bounding DTW in ULISSE.} We now propose a new lower bounding measure for the true DTW distance between a data series and all the sequences (of the same length) represented by an \textit{ULISSE} Envelope. 
To that extent, we first introduce a measure based on $LB_{Keogh}$ distance, which is computed between the PAA representation of $dtwENV_{r}(D)$ and the $iSAX$ representation of $D'$. 
Given $w$, the number of PAA coefficients of each dtw envelope series ($U^{DTW}$,$L^{DTW}$) that is equivalent to the number of iSAX coefficients of $D'$, we have:   

\begin{equation}
\begin{gathered} 
{\small LB_{Keogh_{PAA\_iSAX}}(PAA(dtwENV_{r}(D)),iSAX(D'))=}\\
\sqrt{\frac{|D|}{w}} \sqrt{\sum_{i=1}^{w} 
	\begin{cases}
	\scriptstyle(\beta_{\ell}(iSAX(D')_{i}) - PAA(U^{DTW})_{i} )^2,&  \scriptstyle if (*)\\
	\scriptstyle(PAA(L^{DTW})_{i} - \beta_{u}(iSAX(D')_{i}))^2,& \scriptstyle if (**)\\
	\scriptstyle 0 & {\footnotesize \text{otherwise.}}
	\end{cases}}\\
\scriptstyle(*) \beta_{\ell}(iSAX(D')_{i}) > PAA(U^{DTW})_{i} )\\
\scriptstyle(**) PAA(L^{DTW})_{i} > \beta_{u}(iSAX(D')_{i}) 
\label{lbKeoghPAAISAX}
\end{gathered}
\end{equation}

We know that $LB_{Keogh_{PAA\_iSAX}}(PAA(dtwENV_{r}(D)),$\\$iSAX(D')) \le LB_{Keogh}(dtwENV_{r}(D),D')$ as proven by Keogh et al.~\cite{DBLP:journals/kais/KeoghR05}.
Given the $PAA$ representation of $dtwENV_{r}(D)$ (of $w$ coefficients), and an \textit{ULISSE} Envelope built on $D'$: $uENV_{paaENV_{[D',\ell_{min},\ell_{max},a,\gamma,s]}})$ = $[L,U]$, we define:

\begin{equation}
\begin{gathered} 
{\small LB_{PaL}(PAA(dtwENV_{r}(D)),uENV_{paaENV[D',...]})=}\\
\sqrt{s} \sqrt{\sum_{i=1}^{w} 
	\begin{cases}
	\scriptstyle(\beta_{\ell}(iSAX(L)_{i}) -  PAA(U^{DTW})_{i} )^2,&  \scriptstyle if (*)\\
	\scriptstyle( PAA(L^{DTW})_{i} - \beta_{u}(iSAX(L)_{i}))^2,& \scriptstyle if (**)\\
	\scriptstyle 0 & {\footnotesize \text{otherwise.}}
	\end{cases}}\\
\scriptstyle(*) \beta_{\ell}(iSAX(L)_{i}) > PAA(U^{DTW})_{i} )\\
\scriptstyle(**) PAA(L^{DTW})_{i} > \beta_{u}(iSAX(U)_{i}) 
\label{PalpanasLinardiBound}
\end{gathered}
\end{equation}

\begin{lemma}
	Given two data series $D$ and $D'$, where $\ell_{min} \le |D| \le \ell_{max}$, the distance $LB_{PaL}(PAA(dtwENV_{r}(D)),$\\$uENV_{paaENV[D',\ell_{min},\ell_{max},a,\gamma,s]})$ is always smaller or equal to $DTW(D,D'_{i,|D|})$, for each $i$ such that $a \le i \le a+\gamma+1$ and $|D'|-(i-1) \ge \ell_{min}$.
\end{lemma}

Intuitively, the lemma states that the $LB_{PaL}$ function always provides a measure that is smaller than the true DTW distance between $D$ and each subsequence in $D'$ of the same length, represented by $uENV_{paaENV[D',\ell_{min},\ell_{max},a,\gamma,s]})$.	

\begin{proof}(sketch):
	We want to prove that {\small $$LB_{PaL}(PAA(dtwENV_{r}(D)),uENV_{paaENV[D',\ell_{min},\ell_{max},a,\gamma,s]})$$} is equal to {\small $$\underset{i}{\operatorname{argmin}}\{LB_{Keogh_{PAA\_iSAX}}(PAA(dtwENV_{r}(D)),iSAX(D'_{i,|D|}))\},$$} 
	where $D'_{i,|D|}$ is a subsequence of $D'$ represented by $uENV_{paaENV[D',\ell_{min},\ell_{max},a,\gamma,s]})$.
	The lemma clearly holds if $LB_{PaL}$ yields zero, since the DTW distance between two series is always positive, or equal to zero. 
	We thus test the case, where Equation~\ref{PalpanasLinardiBound} provides a strictly positive value. 
	In the first case, the $i^{th}$ lower iSAX breakpoint of $L$ in the \textit{$ULISSE$} Envelope ($\beta_{\ell}(iSAX(L)_{i})$) is greater than the $i^{th}$ PAA coefficient of the $U^{DTW}$, namely $PAA(U^{DTW})_{i}$. 
	This implies that any other $i^{th}$ iSAX coefficient, which is contained in the \textit{ULISSE} Envelope is necessarily greater than $\beta_{\ell}(iSAX(L)_{i})$ and $PAA(U^{DTW})_{i}$.
	Hence, the Equation~\ref{PalpanasLinardiBound} is equivalent to the smallest value we can obtain from the first branch of $LB_{Keogh_{PAA\_iSAX}}$ computed between each $i^{th}$ iSAX coefficient of the subsequences in $D'$ (represented in the \textit{$ULISSE$} Envelope) to the $i^{th}$ PAA coefficient of $PAA(U^{DTW})$. 
	$LB_{Keogh_{PAA\_iSAX}}$ always yields a value that is smaller or equal to the true $DTW$ distance, with warping window $r$.     
	
	
	The second case is symmetric. 
	Here, the $\beta_{u}(iSAX(L)_{i})$ coefficient is the closest to $PAA(L^{DTW})_{i}$, and greater than any other $i^{th}$ iSAX coefficient of the \textit{$ULISSE$} Envelope.
	Therefore, Equation~\ref{PalpanasLinardiBound} is equivalent to the smallest value we can obtain on the second branch of $LB_{Keogh_{PAA\_iSAX}}$ computed between each $i^{th}$ iSAX coefficient of the subsequences in $D'$ (represented in the \textit{$ULISSE$} Envelope) to the $i^{th}$ coefficient of $PAA(L^{DTW})$.
	\qedsymbol
\end{proof}

In Figure~\ref{lowerboundingDTW}(c), we depict an example that shows the computation of $LB_{PaL}$ between the DTW Envelope that is built around the prefix of $D$ (\textit{153} points) and the \textit{ULISSE} Envelope of the series $D'$. For this latter, the settings are: $a=1$, $\ell_{min}=153$, $\ell_{max}=255$, $\gamma=0$ and $s=51$. 
In the figure, we represent the iSAX coefficients of the \textit{ULISSE} Envelope, with (gray) rectangles delimited by their breakpoints (dashed horizontal lines). 
The coefficients of $PAA(U^{DTW})$ and $PAA(L^{DTW})$ are represented by red and green solid segments.

\subsection{Approximate search}

Similarity search performed on \textit{ULISSE} index relies on Equation~\ref{mindistPAAuENV} (Euclidean distance) and Equation~\ref{PalpanasLinardiBound} (DTW distance) to prune the search space. 
This allows to navigate the tree, visiting the most promising nodes first.
We thus provide a fast approximate search procedure we report in Algorithm~\ref{approximateSearch}. 
In Line~\ref{internalnodepush} (or Line~\ref{internalnodepush2} if DTW distance is used), we start to push the internal nodes of the index in a priority queue, where the nodes are sorted according to their lower bounding distance to the query. 
Note that in the comparison, we use the largest prefix of the query, which is a multiple of the $PAA$ segment length, used at the index building stage (Line~\ref{largestMultiple}). 
Recall that when the search is performed using the DTW measure, the $PAA$ representation of the query is computed on the DTW envelope ($dtwENV_{r}$) of the segment-length multiple that completely contains the query (Line~\ref{largestMultipleDTW}). 
This envelope is composed by two series, which encode the possible warping alignment according the warping window $r$. 
Therefore, the PAA representation is composed by two sets of coefficients, e.g., $PAA(L^{DTW})$ and $PAA(U^{DTW})$, as we depict in Figure~\ref{lowerboundingDTW}.(c).
Then, the algorithm pops the ordered nodes from the queue, visiting their children in the loop of Line~\ref{loopSubtrees}. 
In this part, we still maintain the internal nodes ordered (Lines~\ref{keepOrdered}-\ref{keepOrdered2}). 

As soon as a leaf node is discovered (Line~\ref{iterateLeaf}), we check if its lower bound distance to the query is shorter than the \textit{bsf}. 
If this is verified, the dataset does not contain any data series that are closer than those already compared with the query. 
In this case, the approximate search result coincides with that of the exact search.
Otherwise, we can load the raw data series pointed by the Envelopes in the leaf, which are in turn sorted according to their position, to avoid random disk reads.
We visit a leaf only if it contains Envelopes that represent sequences of the same length as the query. 
Each time we compute either the true Euclidean distance (Line~\ref{EDcompt}) or the true DTW distance (Line~\ref{DTWcompt}), the best-so-far distance (\textit{bsf}) is updated, along with the $R^{a}$ vector. 
Since priority is given to the most promising nodes, we can terminate our visit, when at the end of a leaf visit the $k$ \textit{bsf}'s have not improved (Line~\ref{endVisit}). 
Hence, the vector $R^{a}$ contains the $k$ approximate query answers. 

\begin{algorithm}[!tb]
	{\scriptsize
		\SetAlgoLined
		\label{approximateSearch}
		\KwIn{\textbf{int} $k$, \textbf{float []} $Q$, \textbf{\textit{ULISSE} index} I, \textbf{int} r \tcp{warping window}}
		\KwOut{\textbf{float [$k$][$|Q|$]} $R^{a}$, \textbf{float []} \textit{bsf}}
		\scriptsize
		\textbf{float []}  $Q^*$ $\leftarrow$ $PAA(Q_{1,..,\lfloor|Q| / I.s \rfloor})$\; \label{largestMultiple}
		\textbf{float [][]}  $Q^{*dtw}$ $\leftarrow$ $PAA(dtwENV_{r}(Q_{1,..,\lfloor|Q| / I.s \rfloor}))$\; \label{largestMultipleDTW}
		\textbf{float[k]} \textit{bsf} $\leftarrow \{\infty,...,\infty\}$ \;
		\textbf{PriorityQueue} nodes\;
		
		\ForEach{node \textbf{in} I.root.children()}
		{ 
			\uIf{Euclidean distance search}
			{
				$nodes.push(node,mindist_{ULiSSE}(Q^{*},node))$;\label{internalnodepush}
			}
			\ElseIf{DTW search}
			{
				$nodes.push(node,LB_{PaL}(Q^{*dtw},node))$;\label{internalnodepush2} 
			}		
		} 
		\While{n = nodes.pop()}
		{\label{loopSubtrees}
			\uIf{n.isLeaf() \textbf{and} n.containsSize($|Q|$) }
			{	
				\uIf{n.$lowerBound<$ \textit{bsf}[k] }
				{
					\label{iterateLeaf}
					\tcp{sort according disk pos.}
					\textbf{$\mathbf{uENV}$ []} Envelopes = $sort(n.Envelopes)$\;
					\tcp{iterate the Env. and compute true ED }
					oldBSF $\leftarrow$ \textit{bsf}[k]\;
					\ForEach{E \textbf{in} Envelopes}
					{\textbf{float []} D  $\leftarrow $ readSeriesFromDisk(E)\;
						\For{i $\leftarrow$ E.a \textbf{to} \textbf{min}(E.a+E.$\gamma$+1,$|D|-(|Q|-1)$)}
						{
							\uIf{Euclidean distance search}
							{
								$ED_{updateBSF}$($Q,E.D_{i,|Q|},k,$\textit{bsf}$, R^{a}$);\label{EDcompt}
							}
							\ElseIf{DTW search}
							{
								$DTW_{updateBSF}$($Q,E.D_{i,|Q|},k,$\textit{bsf}$, R^{a}$,$r$);\label{DTWcompt}
							}

						}	
					}
					\tcp{if \textit{bsf} has not improved end visit.}
					\If{oldBSF $==$ \textit{bsf}[k]}
					{\label{endVisit}
						\textit{break}\;
					}
				}
				\Else
				{
					break; \tcp{Approximate search is exact.} 
				}
			}
			\Else
			{
				$LBleft$ $\leftarrow 0$, $LBright$ $\leftarrow 0$\;
				\uIf{Euclidean distance search}
				{
					$LBleft$ $\leftarrow$$mindist_{\textit{ULISSE}}(Q^{*},n.left)$\; 
					$LBright$ $\leftarrow$$mindist_{\textit{ULISSE}}(Q^{*},n.right)$\;
				}
				\ElseIf{DTW search}
				{
					$LBleft$$\leftarrow$$LB_{PaL}(Q^{*dtw},n.left)$\;
					$LBright$$\leftarrow$$LB_{PaL}(Q^{*dtw},n.right)$;
				}
				
				$nodes.push(n.left,LBleft)$\;\label{keepOrdered}
				$nodes.push(n.right,LBright)$\;\label{keepOrdered2}
				
			}
			
		}   
	} 
	\caption{\textit{ULISSE} k-NN-$Approx$}	
\end{algorithm}

\subsection{Exact search}
\begin{algorithm}[!tb]
	{\scriptsize
		\SetAlgoLined
		\label{exactSearch}
		\KwIn{\textbf{int} $k$, \textbf{float []} $Q$, \textbf{\textit{ULISSE} index} I, \textbf{int} r \tcp{warping window}}
		\KwOut{\textbf{float [$k$][$|Q|$]} $R$ }
		\scriptsize
		\textbf{float []}  $Q^*$ $\leftarrow$ $PAA(Q_{1,..,\lfloor|Q| / I.s \rfloor})$\; \label{largestMultiple2}
		\textbf{float [][]}  $Q^{*dtw}$ $\leftarrow$ $PAA(dtwENV_{r}(Q_{1,..,\lfloor|Q| / I.s \rfloor}))$\; \label{largestMultiple2DTW}
		\textbf{float []} \textit{bsf}, \textbf{float [$k$][$|Q|$]} $R$ $\leftarrow$ $k$-$NN$-$Approx(k,Q,I)$ \; 
		\If{\textit{bsf} is not exact}
		{
			\ForEach{E \textbf{in} $I.inMemoryList$}
			{ 
				$LBDist$$\leftarrow 0$\;
				\uIf{Euclidean distance search}
				{
					$LBDist$$\leftarrow$$mindist_{ULiSSE}(Q^{*},E)$;\label{checkLBED}
				}
				\ElseIf{DTW search}
				{
					$LBDist$$\leftarrow$$LB_{PaL}(Q^{*dtw},E)$;\label{checkLBDTW}
				}
				
				\If{ $LBDist$$<$ \textit{bsf}[k]}
				{ \label{checkHardBSF}
					\textbf{float []} D  $\leftarrow $ readSeriesFromDisk(E)\;
					\For{i $\leftarrow$ E.a \textbf{to} \textbf{min}(E.a+E.$\gamma$+1,$|D|-(|Q|-1)$)}
					{\label{iterateCandidates}
						\uIf{Euclidean distance search}
						{
							$ED_{updateBSF}(Q,D_{i,|Q|},k,$\textit{bsf}$, R)$\;
						}
						\ElseIf{DTW search}
						{
							$l \leftarrow  LB_{Keogh}(dtwENV_{r}(Q),D_{i,|Q|})$\;\label{LBKeoghComp}
							\If{$l<$\textit{bsf}[k]}
							{
								$DTW_{updateBSF}(Q,D_{i,|Q|},k,$\textit{bsf}$, R)$\;
							}
						}
					}
				}	
			}
		}
	}
	\caption{\textit{ULISSE} k-NN-$Exact$}	
\end{algorithm}

Note that the approximate search described above may not visit leaves that contain answers better than the approximate answers already identified, and therefore, it will fail to produce exact, correct results.
We now describe an exact nearest neighbor search algorithm, which finds the $k$ sequences with the absolute smallest distances to the query. 

In the context of exact search, accessing disk-resident data following the lower bounding distances order may result in several leaf visits: this process can only stop after finding a node, whose lower bounding distance is greater than the \textit{bsf}, guaranteeing the correctness of the results. 
This would penalize computational time, since performing many random disk I/O might unpredictably degenerate. 

We may avoid such a bottleneck by sorting the Envelopes, and in turn the disk accesses. 
Moreover, we can exploit the \textit{bsf} provided by approximate search, in order to perform a sequential search with pruning over the sorted Envelopes list (this list is stored across the \textit{ULISSE} index).
Intuitively, we rely on two aspects. 
First, the \textit{bsf}, which can translate into a tight-enough bound for pruning the candidate answers.
Second, since the list has no hierarchy structure, any Envelope is stored with the highest cardinality available, which guarantees a fine representation of the series, and can contribute to the pruning process. 

Algorithm~\ref{exactSearch} describes the exact search procedure. 
In the case of Euclidean distance search, in Line~\ref{checkLBED} we compute the lower bound distance between the Envelope and the query.
On the other hand, when DTW distance is used, we compute the lower bound distance in Line~\ref{checkLBDTW}.      
If it is not smaller than the $k^{th}$ \textit{bsf}, we do not access the disk, pruning Euclidean Distance computations as well.
Note that while we are computing the true distances, we can speed-up computations using the \textit{Early Abandoning} technique~\cite{DBLP:conf/kdd/RakthanmanonCMBWZZK12}, which works both for Euclidean and DTW distances.
In the case of DTW distance, prior to computing the raw distance, we have a further possibility to prune computations using the $LB_{Keogh}$ (Equation~\ref{lbKeogh}) in Line~\ref{LBKeoghComp}. 
This permits to obtain a lower bounding measure in linear time, avoiding the full DTW calculation.

\subsection{Complexity of query answering}

We provide now the time complexity analysis of query answering with \textit{ULISSE}. 
Both the approximate and exact query answering time strictly depend on data distribution as shown in~\cite{DBLP:conf/kdd/ZoumpatianosLPG15}. 
We focus on exact query answering, since approximate is part of it.

\noindent{\bf Best Case.} 
In the best case, an exact query will visit one leaf at the stage of the approximate search (Algorithm~\ref{approximateSearch}), and during the second leaf visit will fulfill the stopping criterion (i.e., the \textit{bsf} distance is smaller than the lower bounding distance between the second leaf and the query). 
Given the number of the first layer nodes (root children) $N$, the length of the first leaf path $L$, and the number of Envelopes in the leaf $S$, the best case complexity is given by the cost to iterate the first layer node and descend to the leaf keeping the nodes sorted in the heap: $O(w(N + L log L))$, where $w$ is the number of symbols checked at each lower bounding distance computation. We recall that computing the lower bound of Euclidean or DTW distance has equal time complexity. Moreover we need to take into account the additional cost of sorting the disk accesses and computing the true distances in the leaf, which is $O(S(logS+W))$ in the case of Euclidean distance, and $O(S(logS+ r W))$ for DTW distance, where $W=\ell_{min}(\ell_{max}-\ell_{min}+\gamma+1)$ represents the maximum number of points to check in each Envelope, and $r$ is the warping window length.
Note that we always perform disk accesses sequentially, avoiding random disk I/O. Each disk access in \textit{ULISSE} reads at most $\Theta(\ell_{max}+\gamma)$ points.

\noindent{\bf Worst Case.}
The worst case for exact search takes place when at the approximate search stage, the complete set of leaves that we denote with $T$, need to be visited. 
This has a cost of $O(w(N + T L log L))$ plus the cost of computing the true distances, which takes  $O(T(S(logS+W)))$ for Euclidean distance, or $O(T(SlogS+ SrW))$ for DTW distance, where (as above) $W=\ell_{min}(\ell_{max}-\ell_{min}+\gamma+1)$.
Note though that this worst case is pathological: for example, when all the series in the dataset are the same straight lines (only slightly perturbed).
Evidently, the very notion of indexing does not make sense in this case, where all the data series look the same. 
As we show in our experiments on several datasets, in practice, the approximate algorithm always visits a very small number of leaves.

\noindent{\bf \textit{ULISSE} k-NN Exact complexity.}
So far we have considered the exact k-NN search with regards to Algorithm~\ref{approximateSearch} (approximate search). 
When this algorithm produces approximate answers, providing just an upper bound \textit{bsf}, in order to compute exact answers we must run Algorithm~\ref{exactSearch} (exact search). 
The complexity of this procedure is given by the cost of iterating over the Envelopes and computing the \textit{mindist}, which takes $O(Mw)$ time, where $M$ is the total number of Envelopes in the index. 
Let's denote with $V$ the number of Envelopes, for which the raw data are retrieved from disk and checked. 
Then, the algorithm takes an additional $O(VW)$ time to compute the true Euclidean distances, or $O(VrW)$ to compute the true DTW distances, with $W=\ell_{min}(\ell_{max}-\ell_{min}+\gamma+1)$.

\noindent{\bf $\mathbf{\epsilon}$-range search adaption.}
We note that Algorithm~\ref{exactSearch} can be easily adapted to support $\epsilon$-range search, without affecting its time complexity.
In order to retrieve all answers with distance less than a given threshold $\epsilon \in \mathbb{R}$, we just need to replace the bound $bsf[k]$ with $\epsilon$, in the test of line~\ref{checkHardBSF}.
Subsequently, if the test is true, the algorithm will compute the real distances between the query and all candidates in $D$ (line~\ref{iterateCandidates}), simply filtering the subsequences with distances lower than $\epsilon$.  

\section{Experimental Evaluation}
\label{sec:experiments}

\noindent{\bf Setup.}
All the experiments presented in this section are completely reproducible: the code and datasets we used are available online~\cite{1}.
We implemented 
all algorithms (indexing and query answering) in C (compiled with gcc 4.8.2).
We ran experiments on an Intel Xeon E5-2403 (4 cores @ 1.9GHz), using the x86\_64 GNU/Linux OS environment.

\noindent{\bf Algorithms.}
We compare $\textit{ULISSE}$ to \emph{Compact Multi-Resolution Index (CMRI)~\cite{DBLP:journals/kais/KadiyalaS08}}, which is the current state-of-the-art index for similarity search with varying-length queries (recall that \emph{CMRI} constructs a limited number of distinct indexes for series of different lengths). 
We note though, that in contrast to our approach, \emph{CMRI} can only support non Z-normalized sequences.
Furthermore, we compare $\textit{ULISSE}$ to \emph{KV-Match}~\cite{KV-match}, which is the state-of-the-art indexing technique for $\epsilon$-range queries that support the Euclidean and DTW measures over non Z-normalized sequences (remember that, as we discussed in Section~\ref{sec:relatedwork}, for Z-normalized data \emph{KV-Match} only supports exact search for the constrained $\epsilon$-range queries). 
Finally, we consider Index Interpolation \emph{(IND-INT)}~\cite{DBLP:journals/datamine/LohKW04}. 
This method adapts an index based $\epsilon$-range algorithm, which supports Z-normalization to answer k-NN queries of variable length. 

In addition, we compare to the current state-of-the-art algorithms for subsequence similarity search, the \emph{UCR suite}~\cite{DBLP:conf/kdd/RakthanmanonCMBWZZK12}, and \emph{MASS}~\cite{DBLP:conf/icdm/MueenHE14}. Note that only \emph{UCR suite} works with the Euclidean and DTW measures, whereas \emph{MASS} supports only similarity search using Euclidean distance.
These algorithms do not use an index, but are based on optimized serial scans, and are natural competitors, since they can process overlapping subsequences very fast.

\noindent{\bf Datasets.}
For the experiments, we used both synthetic and real data.
We produced the synthetic datasets with a generator, where a random number is drawn from a Gaussian distribution $N(0,1)$, then at each time point a new number is drawn from this distribution and added to the value of the last number. This kind of data generation has been extensively used in the past~\cite{DBLP:conf/kdd/ZoumpatianosLPG15,DBLP:journals/vldb/ZoumpatianosLIP18}, and has been shown to effectively model real-world financial data~\cite{Faloutsos1994}.

The real datasets we used are:
\begin{itemize}
	\item (GAP), which contains the recording of the global active electric power in France for the period 2006-2008. This dataset is provided by EDF (main electricity supplier in France)~\cite{Lichman:2013};
	\item (CAP), the Cyclic Alternating Pattern dataset, which contains the EEG activity occurring during NREM sleep phase~\cite{CAP:dataset};
	\item (ECG) and (EMG) signals from Stress Recognition in Automobile Drivers~\cite{stressDriver};
	\item (ASTRO), which contains data series representing celestial objects~\cite{ltv};
	\item (SEISMIC), which contains  seismic data series, collected from the IRIS Seismic Data Access repository~\cite{SEISMIC}.
\end{itemize}

In our experiments, we test queries of lengths \textit{160-4096} points, since these cover at least $90\%$ of the ranges explored in works about data series indexing in the last two decades \cite{DBLP:journals/datamine/KeoghK03,DBLP:journals/datamine/BagnallLBLK17,DBLP:journals/datamine/WangMDTSK13}.

\subsection{Envelope Building}

In the first set of experiments, we analyze the performance of the $\textit{ULISSE}$ indexing algorithm. 
Note that the indexing algorithm is oblivious to the distance measure used at query time.

In Figure~\ref{indexingExp}(a) we report the indexing time (Envelope Building and Bulk loading operations) when varying $\gamma$. 
We use a dataset containing \textit{5M} series of length \textit{256}, fixing $\ell_{min}=160$ and $\ell_{max}=256$. 
Observe that, when $\gamma =0$, the algorithm needs to extract as many Envelopes as the number of master series of length $\ell_{min}$. 
This generates a significant overhead for the index building process (due to the maximal Envelopes generation), but also does not take into account the contiguous series of same length, in order to compute the statistics needed for Z-normalization. 
A larger $\gamma$ speeds-up the Envelope building operation by several orders of magnitude, and this is true for a very wide range of $\gamma$ values (Figure~\ref{indexingExp}(a)).
These results mean that the $uENV_{norm}$ building algorithm can achieve good performance in practice, despite its complexity that is quadratic on $\gamma$.

In Figure~\ref{indexingExp}(b) we report an experiment, where $\gamma$ is fixed, and the query length range ($\ell_{max}-\ell_{min}$) varies. 
We use a dataset, with the same size of the previous one, which contains \textit{2.5M} series of length \textit{512}.
The results show that increasing the range has a linear impact on the final running time.

\begin{figure}[tb]
	\hspace*{-0.5cm}
	\includegraphics[trim={0cm 13cm 0cm 3cm},scale=0.55]{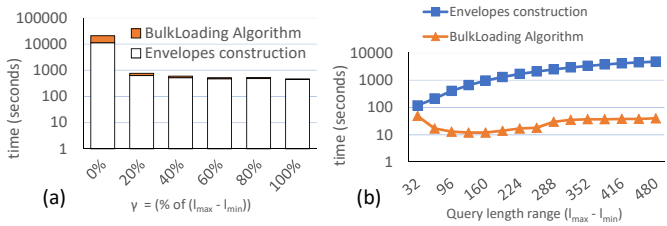}
	\caption{(a) Construction and bulk Loading time (log scale) of Envelopes in 5GB datasets varying $\gamma$ (5M of series of length 256), $\ell_{min}=160$, $\ell_{max}=256$ .
		(b) Construction and Bulk Loading time (log scale) of Envelopes in 5GB dataset (2.5M of series of length 512) varying $\ell_{max} - \ell_{min}$ (lengths range), $\gamma=256$, fixed $\ell_{max}=512$.}
	\label{indexingExp}
\end{figure}

\subsection{Exact Search Similarity Queries with Euclidean Distance}
We now test $\textit{ULISSE}$ on exact 1-Nearest Neighbor queries using Euclidean distance. 
We have repeated this experiment varying the $\textit{ULISSE}$ parameters along predefined ranges, which are (default in bold) $\gamma:[0\%, 20\%, 40\%, 60\%, 80\%, \mathbf{100\%}]$, where the percentage is referring to its maximum value, $\ell_{min}: [96, 128, \mathbf{160}, 192, 224, 256]$, $\ell_{max}: [256]$, dataset series length ($\ell_{S}$): $[\mathbf{256}, 512, 1024, 1536, 2048, 2560]$ and dataset size of $5GB$. 
Here, we use synthetic datasets containing random walk data in binary format, where a single point occupies 4 bytes. Hence, in each dataset $C$, where $|C|^{Bytes}$ denotes the corresponding size in bytes, we have a number of subsequences of length $\ell$ given by $N^{seq} = (\ell_{S} - \ell + 1) \times ((|C|^{Bytes}/4)/\ell_{S})$.
For instance, in a \textit{5GB} dataset, containing series of length $256$, we have $\sim$\textit{500 Million} subsequences of length \textit{160}.

We record the average \textit{CPU time}, \textit{query disk I/O time} (time to fetch
data from disk: Total time - CPU time), and \textit{pruning power} (percentage of the total number of Envelopes in the index that do not need to be read), of \textit{100} queries, extracted from the datasets with the addition of Gaussian noise. For each index used, the \textit{building time} and the relative \textit{size} are reported.
Note that we clear the main memory cache before answering each set of queries.
We have conducted our experiments using datasets that are both smaller and larger than the main memory.

In all experiments, we report the cumulative running time of \textit{1000} random queries for each query length.

\begin{figure}[tb]
	\includegraphics[trim={0cm 3.5cm 10cm 3cm},scale=0.55]{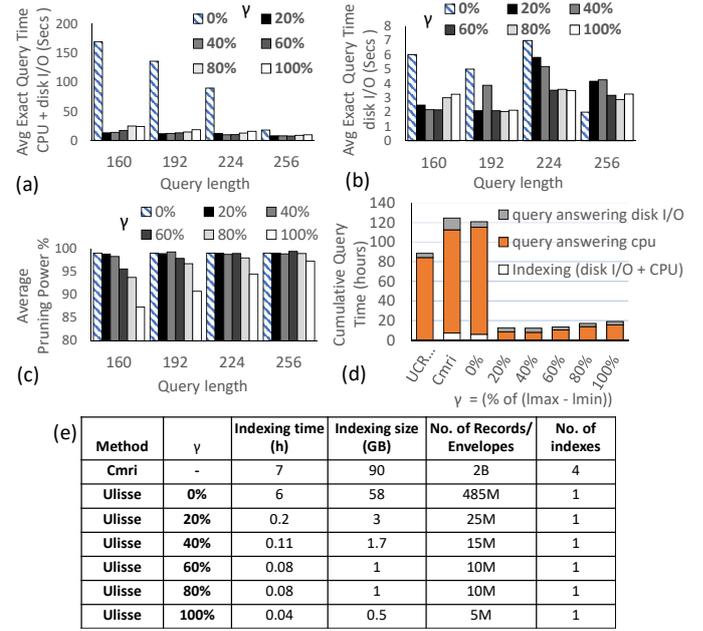}
	\caption{Query answering time performance, varying $\gamma$ on non Z-normalized data series. \textit{a)} $\textit{ULISSE}$ average query time (CPU + disk I/O). \textit{b)} $\textit{ULISSE}$ average query disk I/O time. \textit{c)} $\textit{ULISSE}$ average query pruning power. \textit{d)} Comparison of $\textit{ULISSE}$ to other techniques (cumulative indexing + query answering time). \textit{e)} Table resuming the indexes' properties.}
	\label{1Exp}
\end{figure}

\noindent{\bf Varying $\gamma$.}
We first present results for similarity search queries on $\textit{ULISSE}$ when we vary~$\gamma$, ranging from \textit{0} to its maximum value, i.e., $\ell_{max}-\ell_{min}$. 
In Figure~\ref{1Exp}, we report the results concerning non Z-normalized series (for which we can compare to \textit{CMRI}).
We observe that grouping contiguous and overlapping subsequences under the same summarization (Envelope) by increasing $\gamma$, affects positively the performance of index construction, as well as query answering (Figures~\ref{1Exp}(a) and~(d)).
The latter may seem counterintuitive, since $\gamma$ influences in a negative way pruning power, as depicted in Figure~\ref{1Exp}(c). 
Indeed, inserting more master series into a single $Envelope$ is likely to generate large containment areas, which are not tight representations of the data series.
On the other hand, it leads to an overall number of $Envelopes$ that is several orders of magnitude smaller than the one for $\gamma=0\%$. 
In this last case, when $\gamma=0$, the algorithm inserts in the index as many records as the number of master series present in the dataset (\textit{485}M), as reported in (Figure~\ref{1Exp}(e)).

We note that the disk I/O time on compact indexes is not negatively affected at the same ratio of pruning power. 
On the contrary, in certain cases it becomes faster. 
For example, the results in Figure~\ref{1Exp}(b) show that for query length 160, the $\gamma=100\%$ index is more than 2x faster in disk I/O than the $\gamma=0\%$ index, despite the fact that the latter index has an average pruning power that is $14\%$ higher (Figure~\ref{1Exp}(c)). 
This behavior is favored by disk caching, which translates to a higher hit ratio for queries with slightly larger disk load. 
We note that we repeated this experiment several times, with different sets of queries that hit different disk locations, in order to verify this specific behavior. The results showed that this disk I/O trend always holds. 

While disk I/O represents on average the $3-4$\% of the total query cost, computational time significantly affects the query performance. 
Hence, a compact index, containing a smaller number of $Envelopes$, permits a fast in memory sequential scan, performed by Algorithm~\ref{exactSearch}.

In Figure~\ref{1Exp}(d) we show the cumulative time performance (i.e., $4,000$ queries in total), comparing $\textit{ULISSE}$, \textit{CMRI}, and \textit{UCR Suite}. 
Note that in this experiment, $\textit{ULISSE}$ indexing time is negligible w.r.t. the query answering time. 
$\textit{ULISSE}$, outperforms both \textit{UCR Suite} and \textit{CMRI}, achieving a speed-up of up to $12x$.

Further analyzing the performance of \textit{CMRI}, we observe that it constructs four indexes (for four different lengths), generating more than $2B$ index records.
Consequently, it is clear that the size of these indexes will negatively affect the performance of \textit{CMRI}, even if it achieves reasonable pruning ratios.

These results suggest that the idea of generating multiple copies of an index for different lengths, is not a scalable solution.

In Figure~\ref{1Exp_normalized}, we show the results of the previous experiment, when using Z-normalization. 
We note that in this case the query answering time has an overhead generated by the Z-normalization that is performed on-the-fly, during the similarity search stage.
Overall, we observe exactly the same trend as in non Z-normalized query answering. 
$\textit{ULISSE}$ is still $2x$ faster than the state-of-the-art, namely \textit{UCR Suite}.

\begin{figure}[tb]
	\includegraphics[trim={0cm 3.5cm 10cm 3cm},scale=0.55]{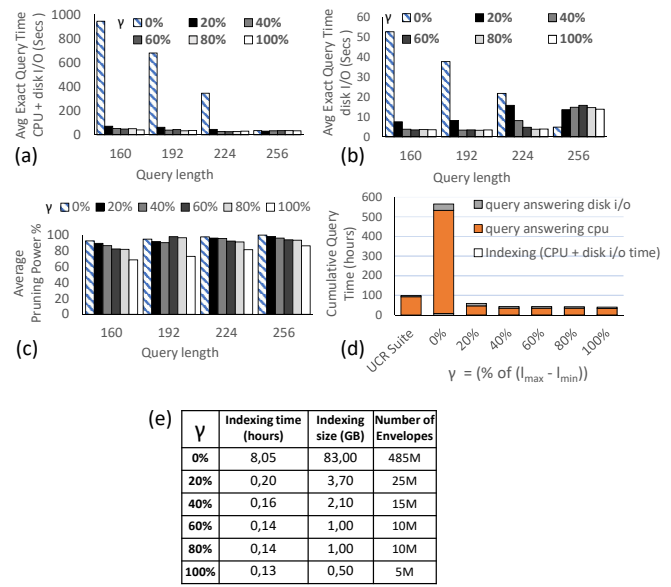}
	\vspace*{-0.4cm}
	\caption{$\textit{ULISSE}$ Indexing and exact queries performance on Z-normalized series, varying sigma. \textit{a)} Average query time (CPU + disk I/O). \textit{b)} Average query disk I/O time. \textit{c)} Average query pruning power. \textit{d)} Comparison to other techniques (cumulative indexing + query answering time). \textit{e)} Table resuming the indexes' properties.}
	\vspace*{-0.4cm}
	\label{1Exp_normalized}
\end{figure}

\begin{figure*}[tb]
	\hspace*{-1cm}
	\centering
	\includegraphics[trim={0cm 7.8cm 2cm 3cm},scale=0.60]{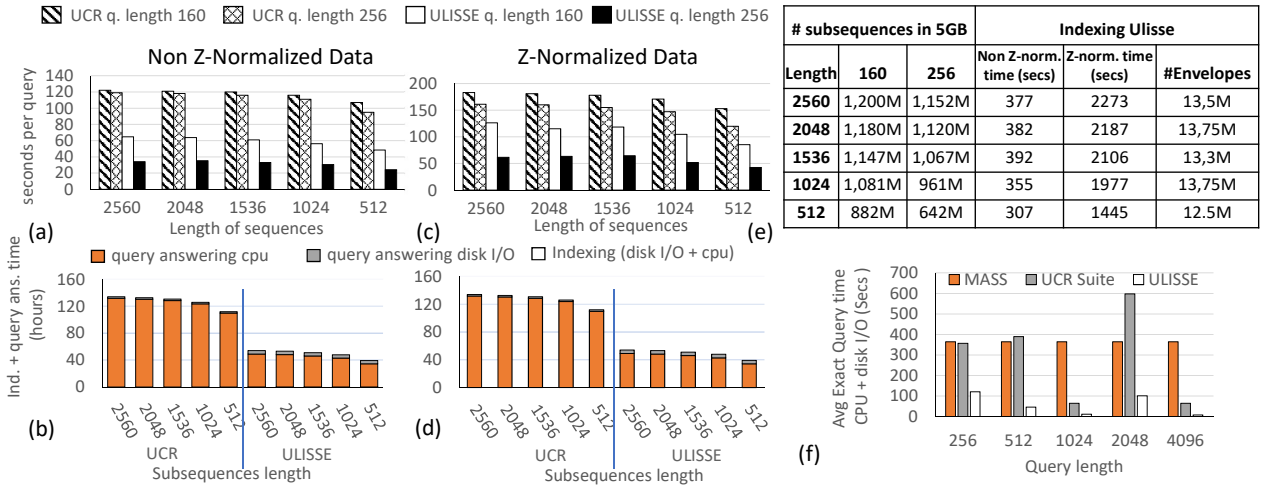}
	\caption{Query answering time performance of $\textit{ULISSE}$ and \textit{UCR Suite}, varying the data series size. Average query (CPU time + disk I/O) (\textit{a}) for non Z-normalized, (\textit{c}) for Z-normalized series). Cumulative indexing + query answering time (\textit{b}) for non Z-normalized, (\textit{d}) for Z-normalized series). \textit{e)} Table resuming the indexes' properties. \textit{f)} Comparison between MASS algorithm, \textit{UCR Suite} and $\textit{ULISSE}$.}
	\label{VarLengthSeriesDataset}
\end{figure*}

\noindent{\bf Varying Length of Data Series.}
In this part, we present the results concerning the query answering performance of $\textit{ULISSE}$ and \textit{UCR Suite}, as we vary the length of the sequences in the indexed datasets,
as well as the query length (refer to Figure~\ref{VarLengthSeriesDataset}).
In this case, varying the data series length in the collection, leads to a search space growth, in terms of overlapping subsequences, as reported in Figure~\ref{VarLengthSeriesDataset}(e). 
This certainly penalizes index creation, due to the inflated number of Envelopes that need to be generated.
On the other hand, \textit{UCR Suite} takes advantage of the high overlapping of the subsequences during the in-memory scan. 
Note that we do not report the results for $CMRI$ in this experiment, since its index building time would take up to \textit{1 day}. 
In the same amount of time, $\textit{ULISSE}$ answers more than $1,000$ queries.

Observe that in Figures~\ref{VarLengthSeriesDataset}(a) and~(c), $\textit{ULISSE}$ shows better query performance than the UCR suite, growing linearly as the search space gets exponentially larger. This demonstrates that $\textit{ULISSE}$ offers a competitive advantage in terms of pruning the search space that eclipses the pruning techniques \textit{UCR Suite}. 
The aggregated time for answering $4,000$ queries ($1,000$ for each query length) is 2x for $\textit{ULISSE}$ when compared to \textit{UCR Suite} (Figures~\ref{VarLengthSeriesDataset}(b) and~(d)).


\noindent{\bf Comparison to Serial Scan Algorithms using Euclidean Distance.} We now perform further comparisons to serial scan algorithms, namely, \emph{MASS} and \textit{UCR Suite}, with varying query lengths.

\emph{MASS}~\cite{DBLP:conf/icdm/MueenHE14} is a recent data series similarity search algorithm that
computes the distances between a Z-normalized query of length $l$ and all the Z-normalized overlapping subsequences of a single sequence of length $n \ge l$. 
\emph{MASS} works by calculating the dot products between the query and $n$ overlapping subsequences in frequency domain, in $logn$ time, which then permits to compute each Euclidean distance in constant time. 
Hence, the time complexity of \emph{MASS} is $O(nlogn)$, and is independent of the data characteristics and the length of the query ($l$).
In contrast, the \textit{UCR Suite} effectiveness of pruning computations may be significantly affected by the data characteristics.

We compared $\textit{ULISSE}$ (using the default parameters), MASS and \textit{UCR Suite} on a dataset containing \textit{5M} data series of length \textit{4096}. 
In Figure~\ref{VarLengthSeriesDataset}(f), we report the average query time (CPU + disk/io) of the three algorithms.

We note that MASS, which in some cases is outperformed by \textit{UCR Suite} and $\textit{ULISSE}$, is strongly penalized, when ran over a high number of non overlapping series. 
The reason is that, although MASS has a low time complexity of $O(nlogn)$, the Fourier transformations (computed on each subsequence) have a non negligible constant time factor that render the algorithm suitable for computations on very long series.

\noindent{\bf Varying Range of Query Lengths.} In the last experiment of this subsection, we investigate how varying the length range [$\ell_{min};\ell_{max}$] affects query answering performance. 

\begin{figure}[tb]
	\includegraphics[trim={0cm 8.5cm 10cm 3cm},scale=0.55]{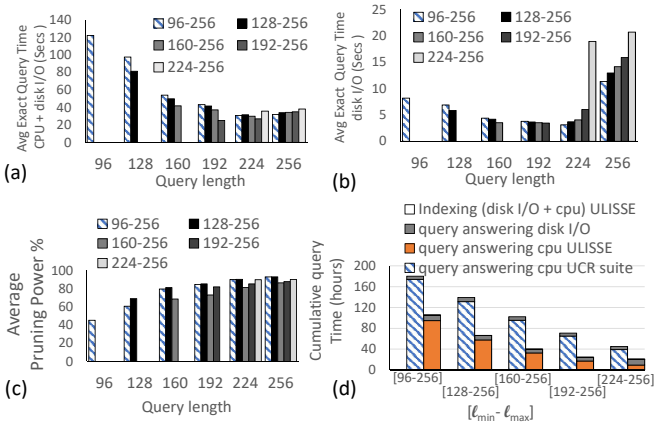}
	\caption{Query answering time, varying the range of query length on Z-normalized data series. (\textit{a)} $\textit{ULISSE}$ average query time (CPU + disk I/O). (\textit{b)} $\textit{ULISSE}$ average query disk I/O time. (\textit{c)} $\textit{ULISSE}$ average query pruning power. (\textit{d)} $\textit{ULISSE}$ comparison to other techniques  (cumulative indexing + query answering time).}
	\label{varRangeQueryNorm}
\end{figure}

In Figure~\ref{varRangeQueryNorm}, we depict the results for Z-normalized sequences. 
We observe that enlarging the range of query length, influences the number of Envelopes we need to accommodate in our index.
Moreover, a larger query length range corresponds to a higher number of Series (different normalizations), which the algorithms needs to consider for building a single Envelope (loop of line~\ref{loopNonMasterandMasterSeries} of Algorithm~\ref{algoEnvNorm}). 
This leads to large containment areas and in turn, coarse data summarizations. 
In contrast, Figure~\ref{varRangeQueryNorm}(c) indicates that pruning power slightly improves as query length range increases. 
This is justified by the higher number of Envelopes generated, when the query length range gets larger. 
Hence, there is an increased probability to save disk accesses. 
In Figure~\ref{varRangeQueryNorm}(a) we show the average query time (CPU + disk I/O) on each index, observing that this latter is not significantly affected by the variations in the length range.
The same is true when considering only the average query disk I/O time (Figure~\ref{varRangeQueryNorm}(b)), which accounts for $3-4$\% of the total query cost. We note that the cost remains stable as the query range increases, when the query length varies between \textit{96}-\textit{192}. For queries of length \textit{224} and \textit{256}, when the range is the smallest possible the disk I/O time increases.
This is due to the high pruning power, which translates into a higher rate of cache misses.  
In Figure~\ref{varRangeQueryNorm}(d), the aggregated time comparison shows $\textit{ULISSE}$ achieving an up to $2x$ speed-up over \textit{UCR Suite}.

\begin{figure}[tb]
	\includegraphics[trim={0cm 8.5cm 10cm 3cm},scale=0.55]{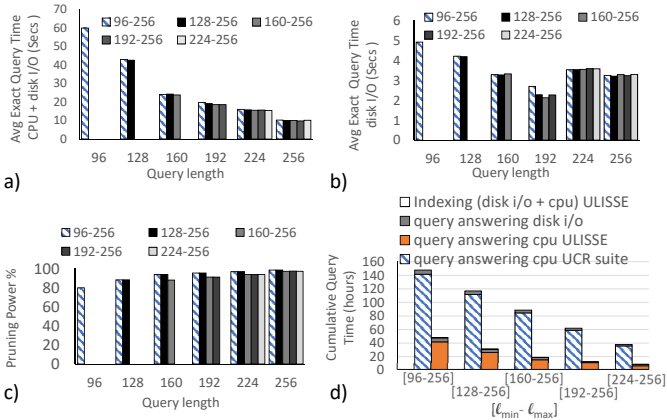}
	\caption{Query answering time, varying the range of query length on non Z-normalized data series. (\textit{a)} $\textit{ULISSE}$ average query time (CPU + disk I/O). (\textit{b)} $\textit{ULISSE}$ average query disk I/O time. (\textit{c)} $\textit{ULISSE}$ average query pruning power. (\textit{d)} $\textit{ULISSE}$ comparison to other techniques  (cumulative indexing + query answering time).}
	\label{varRangeQuery}
\end{figure}

In Figure~\ref{varRangeQuery} we present the results for non Z-normalized sequences, where the same observations hold.
Moreover, as we previously mentioned, when Z-normalization is not applied the pruning power slightly increases. 
This leads ULISSE to a performance up to $3x$ faster than \textit{UCR Suite}.

\subsection{Approximate Search Queries}

\begin{figure}[tb]
	\centering
	\includegraphics[trim={0cm 13cm 11cm 3cm},scale=0.57]{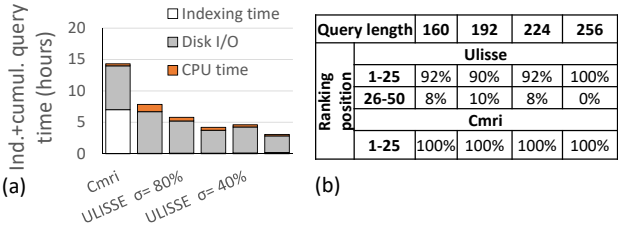}
	\caption{Approximate query answering on non Z-normalized data series. (\textit{a)} Cumulative Indexing + approximate search query time (CPU + disk I/O) of $4,000$ queries ($1,000$ per each query length in [160,192,224,256]). (\textit{b)} Approximate quality: percentage of answers in the relative exact search range.}
	\label{cmriUlisseApprox}
\end{figure}

\noindent{ \bf Approximate Search with Euclidean Distance.} In this part, we evaluate $\textit{ULISSE}$ approximate search. 
Since we compare our approach to CMRI, Z-normalization is not applied. 
Figure~\ref{cmriUlisseApprox}(a) depicts the cumulative query answering time for $4,000$ queries. 
As previously, we note that the indexing time for $\textit{ULISSE}$ is relatively very small. 
On the other hand, the time that CMRI needs for indexing is 2x more than the time during which $\textit{ULISSE}$s has finished indexing and answering $4,000$ queries.

In Figure~\ref{cmriUlisseApprox}(b), we measure the quality of the Approximate search. 
In order to do this, we consider the exact query results ranking, showing how the approximate answers are distributed along this rank, which represents the ground truth. 
We note that CMRI answers have better positions than the $\textit{ULISSE}$ ones. 
This happens thanks to the tighter representation generated by the complete sliding window extraction of each subsequence, employed by CMRI. 
Nevertheless, this small penalty in precision is balanced out by the considerable time performance gains: $\textit{ULISSE}$ is up to 15x faster than CMRI. 
When we use a smaller $\gamma$, (e.g., $20$), $\textit{ULISSE}$ shows its best time performance.
This is due to tighter $Envelopes$ containment area, which permits to find a better best-so-far with a shorter tree index visit.


\begin{figure}[tb]
	\centering
	\includegraphics[trim={0cm 3cm 4cm 3cm},scale=0.50]{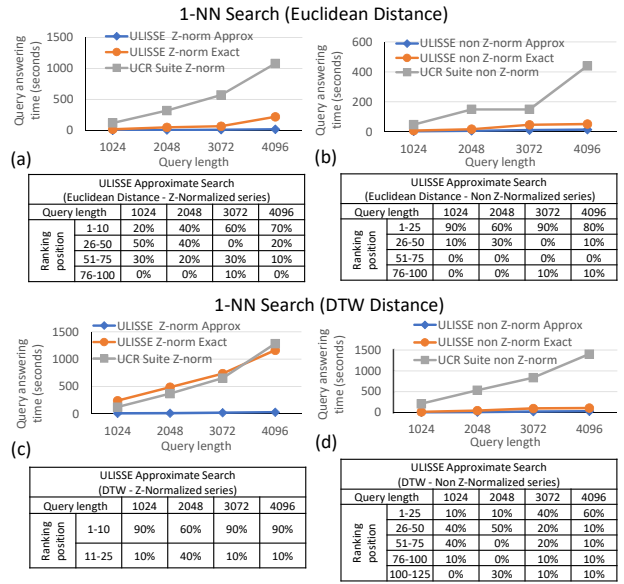}
	\caption{ Average query answering and approximate quality varying query length. (\textit{a)} Z-normalized search with Euclidean distance.
		(\textit{b)} Non Z-normalized search with Euclidean distance.
		(\textit{c)} Z-normalized search with DTW measure.
		(\textit{d)} Non Z-normalized search with DTW measure. }
	\label{UlisseEDDTW_LongApprox}
\end{figure}

\noindent{ \bf Approximate Search with DTW.} Here we evaluate, the time performance of query answering, along with the quality of approximate search.
We test the search using both the Euclidean and DTW measures, on a synthetic series composed of $100M$ points.
We test a query length range between $\ell_{min} = 1024$ and $\ell_{max} = 4096$. 
The other parameters are set to their default value.

In Figures~\ref{UlisseEDDTW_LongApprox}(a) and~(b), we report the average query answering time for the Z-normalized and non Z-normalized cases, respectively.
The results show that  $\textit{ULISSE}$ answers queries up to one order of magnitude faster than \textit{UCR Suite}. 
Furthermore, we note that $\textit{ULISSE}$ scales better as the query length increases.
This shows that our pruning strategy over summarized data, as well as having a good \emph{bsf} approximate answer early on, represent a concrete advantage when pruning the search space.

In Figures~\ref{UlisseEDDTW_LongApprox}(c) and~(d), we report the time performance of query answering with the DTW measure, considering both Z-normalized and non Z-normalized search.
In Figure~\ref{UlisseEDDTW_LongApprox}(c), we observe that $\textit{ULISSE}$ answers queries slightly slower than \textit{UCR Suite}, for three of the query lengths.
This behavior is explained by the fact that the (overlapping) subsequences represented by the Envelopes have a total size $\sim43$x bigger than the original data points.
In this case, the pruning power does not mitigate this disadvantage. 

Overall, the results show that \textit{ULISSE} is a scalable solution.
Moreover, the approximate search, which in this experiment does not visit more than $5$ leaves in the tree, represents a very fast solution, approximating well the exact answer (refer to the tables below each plot of Figure~\ref{UlisseEDDTW_LongApprox}). 
We observed the same trend of visited leaves in all the experiments presented in this work. 
This means that in practice the time complexity of the Approximate search is very close to its best case having a constant query answering complexity.


\subsection{Experiments with Real Datasets}

In this part, we discuss the results of indexing and query answering performed on real datasets. 
Here, we also consider the use of the Dynamic Time Warping (DTW) distance measure, along with Euclidean distance.    

We start the evaluation by considering five different real datasets that fit the main memory. 
In the next sections, we will additionally consider real data series collections that do not fit in the available main memory.
The objective of this experiment is to firstly assess the benefit of maximizing the number of subsequences represented by a $\textit{ULISSE}$ Envelope on query answering time. 
Moreover, we want to analyze the impact of the DTW measure on query time performance.


\noindent{\bf Indexing.} 
For this experiment, we used five real dataset, where each one contains \textit{500K} data series of length 256 (ASTRO, EMG, EEG, ECG, GAP). 
We show in Figure~\ref{indexing_ED_real}.(a,b) the indexing time performance, varying $\gamma$ for both non Z-Normalized and Normalized sequences. 
Recall that $\gamma$ is expressed as the percentage of the maximum number of master series that is $\ell_{max}-\ell_{min}$. 
The results confirm the trend depicted in Figure~\ref{indexingExp}, where the time of building $\textit{ULISSE}$ Envelopes that contain all the master series of each series is one order of magnitude smaller than the time of building the most compact Envelopes, obtained with $\gamma=5\%$. We also note that the overhead generated by the Z-normalization operations, which have an additional $\gamma$ factor in the time complexity of the indexing algorithm, is amortized by the generation of $\sim20$x less Envelopes in the index, as depicted in Figure~\ref{indexing_ED_real}(c).

\begin{figure}[tb]
	\centering
	\includegraphics[trim={1cm 13,5cm 13cm 3cm},scale=0.53]{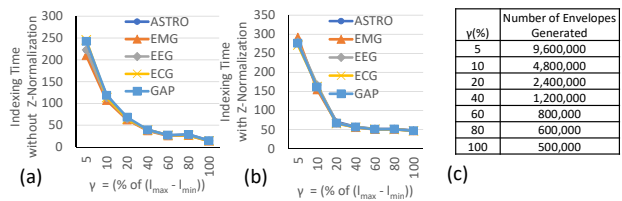}
	\caption{Indexing time of five real datasets (ASTRO, EMG, EEG, ECG, GAP) varying the number of master series in the Envelope ($\gamma$). The datasets contain $500K$ data series of length $256$, whereas $\ell_{min}=160$ and $\ell_{max}=256$. \textit{(a)} Indexing of non Z-Normalized series. \textit{(b)} Indexing time of Z-Normalized series. }
	\label{indexing_ED_real}
\end{figure}

\noindent{\bf Query Answering with Euclidean Distance.} 
We report in Figure~\ref{ED_real} the results obtained for \emph{1-NN} search over Z-normalized sequences, with Euclidean distance. 
All parameters are set to their default values. 
Therefore, in these experiments we used queries of length between $\ell_{min}=160$ and $\ell_{max}=256$; the series in the datasets have length $256$.
In Figure~\ref{ED_real}(a), we report the query pruning power as the number of master series ($\gamma$) in each Envelope varies.
As expected, we can prune less candidates when the Envelopes contain more sequences. 
Recall that when a candidate (subsequence) is pruned, the search does not consider its raw values, thus avoiding both Z-normalization and Euclidean distance computations. 
If a candidate is not pruned, the search can abandon the computations earlier, when the running Euclidean Distance is greater than the $k^{th}$ \textit{bsf} distance.

In Figure~\ref{ED_real}(b), we report the average \emph{abandoning power}, which measures the percentage of the total number of real Euclidean distance computations that are \emph{not} performed. 
When the search processes an increased number of overlapping subsequences, we expect a decrease in the number of computations performed. 
We note that the search avoids computations when the Envelopes contain a large number of subsequences, namely, as $\gamma$ increases. 

In Figure~\ref{ED_real}(c), we report the average query time varying $\gamma$. 
We obtain the highest speed-up, with the most compact index (largest $\gamma$ value), which is more that $2$x faster than the state-of-the-art (\textit{UCR Suite} algorithm). 
This confirms the trend we observed in the previous results conducted over synthetic data. 
We report the average query time for each dataset in Figure~\ref{ED_real}(d), and for each query length in Figure~\ref{ED_real}(e).
In Figure~\ref{ED_real}(f), we show the average number of Euclidean distance and lower bound computations performed by $\textit{ULISSE}$ ($\gamma=100\%$) and \textit{UCR Suite}, as the query length varies (this corresponds to the average number of points on which the distance to the query is computed), as well as the number of points that are loaded from disk and Z-normalized (this corresponds to the overhead generated by the Z-normalization operations).
The goal of this experiment is to quantify the overall benefit of $\textit{ULISSE}$ pruning and abandoning power.
(Recall that \textit{UCR Suite} does not perform any lower bound distance computations when using the Euclidean distance.)

First, we observe that $\textit{ULISSE}$ performs half of the Euclidean distance computations of \textit{UCR Suite}, and considers up to seven time less points for the Z-normalization phase.
Furthermore, we note that the computation of lower bound distances has a negligible impact on the query workload, especially when the query length is smaller than the length of the series in the dataset ($256$), in which case the number of candidate subsequences can be orders of magnitude more. 

\begin{figure}[tb]
	\centering
	\includegraphics[trim={0cm 4cm 2cm 3cm},scale=0.55]{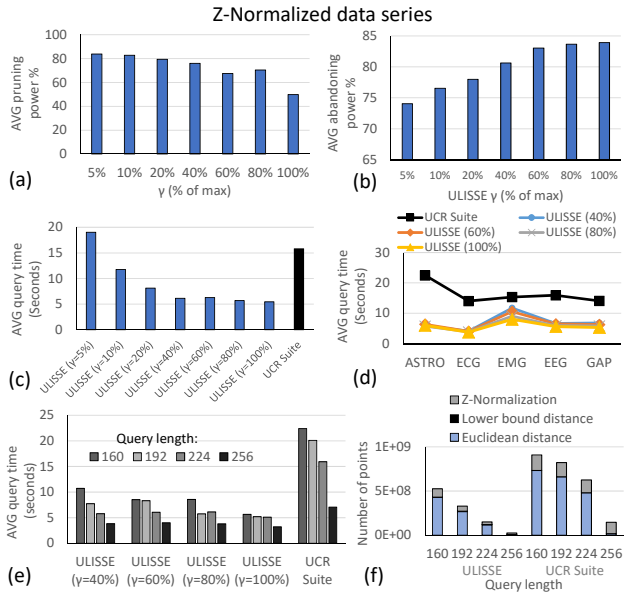}
	\caption{Exact (Z-normalized) query answering with Euclidean distance on real datasets.
		\textit{(a)} Average Pruning power \textit{(b)} Average Abandoning power \textit{(c)} Average query answering time \textit{(d)} Average query answering time for each dataset \textit{(e)} Average query answering time for each query length \textit{(f)} Average query workload (number of points, with $\gamma=100\%$)}
	\label{ED_real}
\end{figure}

In Figure~\ref{ED_real_nonNorm}, we depict the results of query answering, without the use of Z-normalization.
In this case, the results exhibit a small difference in terms of absolute pruning power values, which is higher when the search is performed on absolute series values.
The average query answering time maintains the same trend we observe in Z-normalized query answering. 
On average, $\textit{ULISSE}$ has a $3$x speed-up factor when compared to \textit{UCR Suite}.

\begin{figure}[tb]
	\centering
	\includegraphics[trim={0cm 4cm 2cm 3cm},scale=0.55]{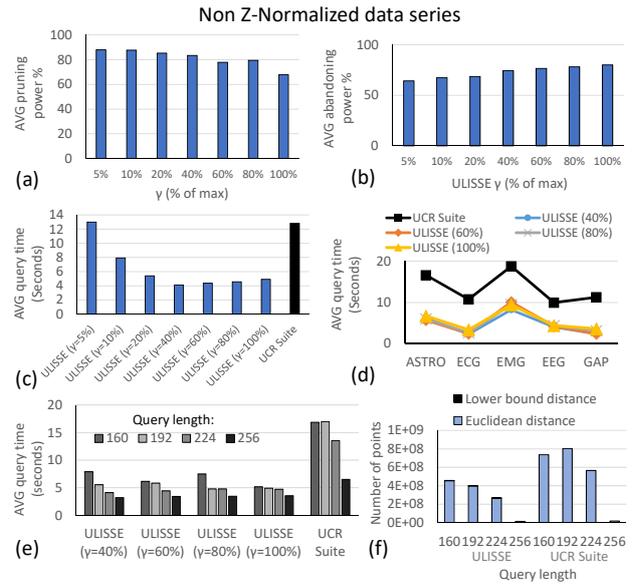}
	\caption{Exact (non Z-normalized) query answering with Euclidean distance on real datasets.
		\textit{(a)} Average Pruning power \textit{(b)} Average Abandoning power \textit{(c)} Average query answering time \textit{(d)} Average query answering time for each dataset \textit{(e)} Average query answering time for each query length \textit{(f)} Average query workload (number of points, with $\gamma=100\%$)}
	\label{ED_real_nonNorm}
\end{figure}

\noindent{\bf Query Answering with DTW Distance.} 
We now report the results of query answering using the DTW measure (Figure~\ref{DTW_real}). 
For this experiment, we used the default parameter settings, and the same real datasets considered in the previous two experiments.
We study the efficiency of query answering (\emph{1-NN} query), which uses the DTW lower bounding measures to prune the search space.

In Figure~\ref{DTW_real}(a) we report the average pruning power, when varying the DTW warping windows from $5\%$ to $15\%$ of the subsequence length.
(These values for the warping window have commonly been used in the literature~\cite{DBLP:journals/kais/KeoghR05}.)
We vary $\gamma$ between $60\%$ and $100\%$ of its maximum value, which give the best running time in this experiment. 
To avoid an unnecessary overload in the plot, we omit the results for $\gamma$ smaller than $60\%$.

Once again, we note that the pruning power is negatively affected by the size of the Envelope ($\gamma$), and under DTW search the abandoning power slightly decreases as the $gamma$ and the $warping$ window get larger (see Figure~\ref{DTW_real}(b)). 
This suggests that the DTW lower bound measure we propose is more sensitive than the one used for Euclidean Distance.
Nevertheless, in the worst case $\textit{ULISSE}$ is still able to prune $20\%$ of the candidates, and to abandon more than $80\%$ of the $DTW$ computations on raw values.

In Figure~\ref{DTW_real}(c) we report the average query answering time varying $\gamma$, and in Figures~\ref{DTW_real}(d) and~(e) the average time for each dataset and for different query lengths, respectively, for $\gamma=100\%$.
For these last two experiments, we observe no significant difference for the other values of $\gamma$ we tested.

We first note that, despite the loss of pruning power of ULISSE when increasing $\gamma$, the query answering time is not significantly affected (refer to Figures~\ref{DTW_real}(c) and~(e)). 
As in the case of Euclidean distance search, the compactness of the $\textit{ULISSE}$ index plays a fundamental role in determining the query time performance, along with the pruning and abandoning power.  

In Figure~\ref{DTW_real}(d), we note that only in the ECG and GAP datasets, enlarging the warping window has a substantial negative effect on query time ($2x$ slower), whereas in the other datasets, and in the worst case the time loss is equivalent to $10\%$.          

In Figure~\ref{DTW_real}(e), we report the average query workload of $\textit{ULISSE}$ and \textit{UCR Suite}. 
In contrast to Euclidean distance queries, we notice that the largest amount of work corresponds to lower bounding distance computations.
Recall that $\textit{ULISSE}$ prunes the search space in two stages: first comparing the query and the data in their summarized versions using $LB_{PaL}$ (Equation~\ref{PalpanasLinardiBound}), and then computing in linear time the $LB_{Keogh}$ between the query and the non pruned candidates.
In the worst case, the DTW distance point-wise computation are 10\% of those performed for calculating the Lower Bound (query length $160$).
In general, the total number of points considered for the whole workload is up to $5$x smaller than for \textit{UCR Suite}. 
We note that the pruning strategy of \textit{UCR Suite} is still very competitive, since it avoids a high number of true distance computations using the $LB_{Keogh}$ lower bound.
Nonetheless, it has to compute the lower bound distance on the entire set of candidates.
The pruning strategy implemented in $\textit{ULISSE}$ permits to achieve up to $10x$ speedup over \textit{UCR Suite}.  

\begin{figure}[tb]
	\centering
	\includegraphics[trim={0cm 3cm 2cm 3cm},scale=0.50]{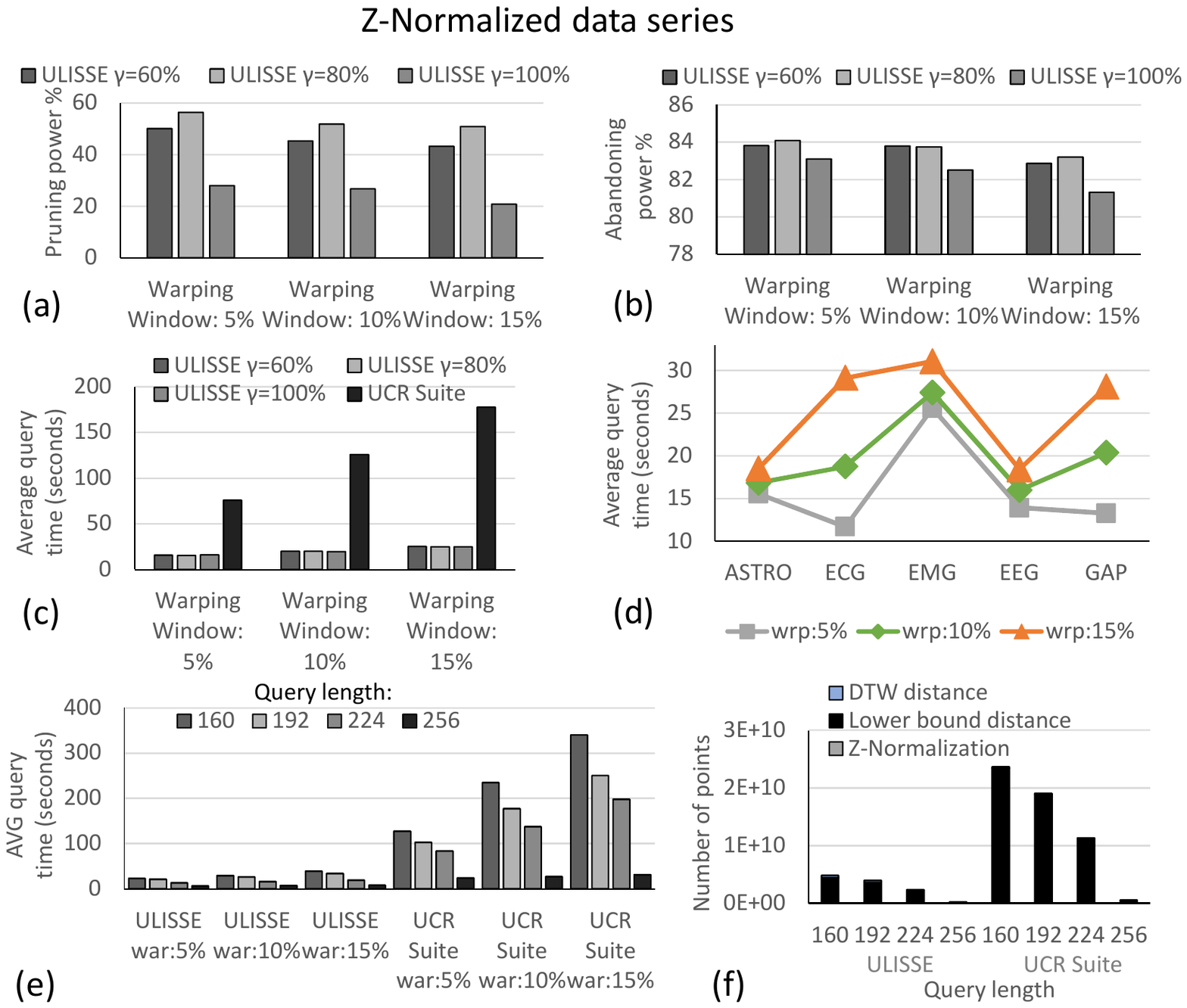}
	\caption{Exact (Z-normalized) query answering with DTW measure on real datasets.
		\textit{(a)} Average Pruning power (varying the warping window) \textit{(b)} Average Abandoning power (varying the warping window) \textit{(c)} Average query answering time \textit{(d)} Average query answering time for each dataset \textit{(e)} Average query answering time for each query length \textit{(f)} Average query workload (number of points, with $\gamma=100\%$)}
	\label{DTW_real}
\end{figure}

In Figure~\ref{DTW_real_nonNorm}, we report the results of DTW search, without the application of Z-Normalization.
Also in this case, we note that the average pruning power of $\textit{ULISSE}$ is higher than the one we previously observed in the Z-normalized search (Figure~\ref{DTW_real_nonNorm}(a)). 
On the other hand, the average abandoning power is less effective, as shown in Figure~\ref{DTW_real_nonNorm}(b).
As a consequence, we can see that the $\textit{ULISSE}$ search performs more DTW distance computations (refer to Figure~\ref{DTW_real_nonNorm}(c)). 
Nevertheless, Figure~\ref{DTW_real_nonNorm}(e) shows that on average $\textit{ULISSE}$ is up to $10$x faster than \textit{UCR Suite}, for all query lengths we tested.

\begin{figure}[tb]
	\centering
	\includegraphics[trim={0cm 3.5cm 2cm 3cm},scale=0.50]{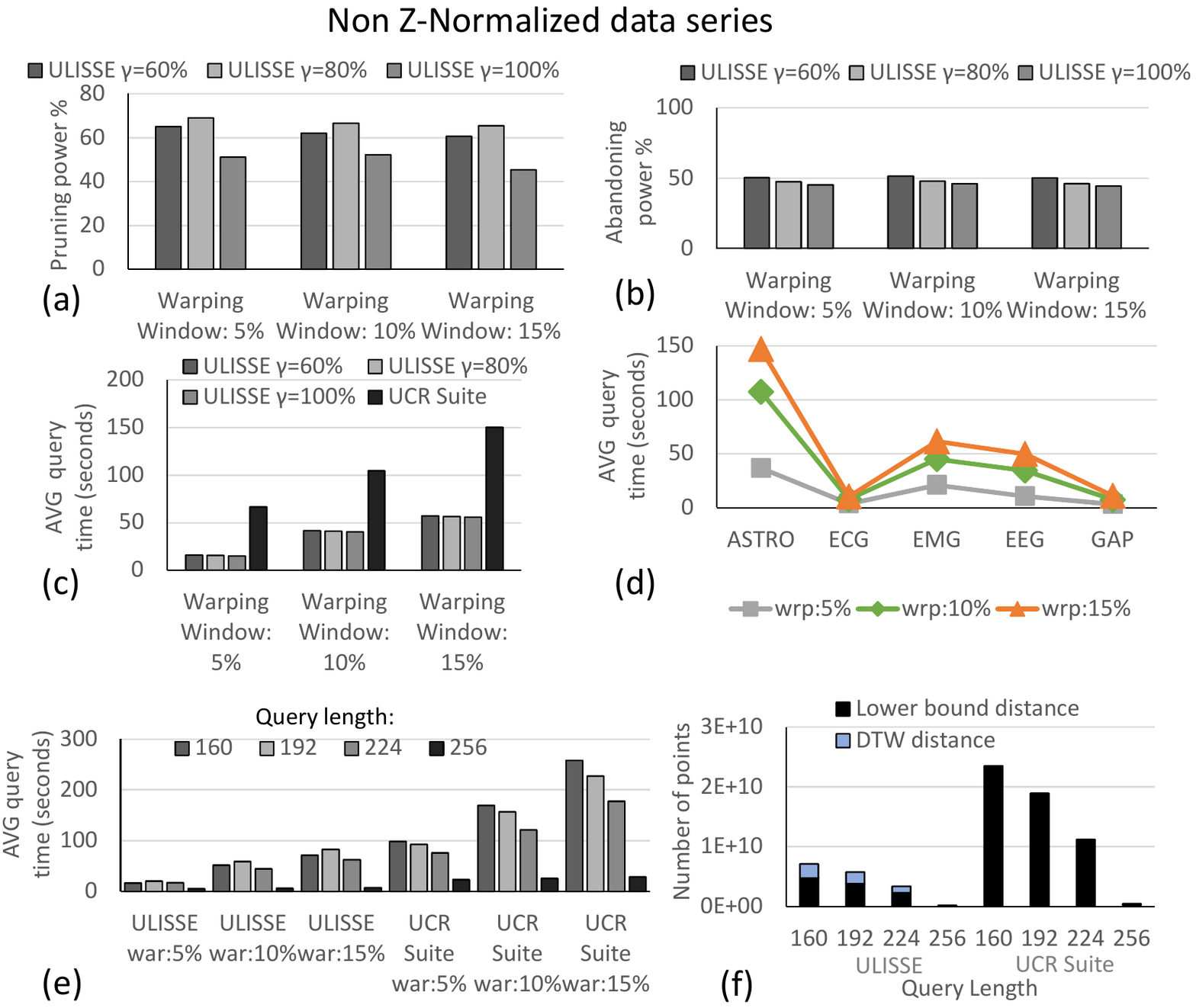}
	\caption{Exact (Z-normalized) query answering with DTW measure on real datasets.
		\textit{(a)} Average Pruning power (varying the warping window) \textit{(b)} Average Abandoning power (varying the warping window) \textit{(c)} Average query answering time \textit{(d)} Average query answering time for each dataset \textit{(e)} Average query answering time for each query length \textit{(f)} Average query workload (number of points, with $\gamma=100\%$)}
	\label{DTW_real_nonNorm}
\end{figure}

\noindent{\bf Query over Large datasets with Euclidean Distance.} 
Here, we test $\textit{ULISSE}$ on three large synthetic datasets of sizes \textit{100}GB, \textit{500}GB, and \textit{750}GB, as well as on two real series collections, i.e., ASTRO and SEISMIC (described earlier). The other parameters are the default ones.
For each generated index and for the \textit{UCR Suite}, we ran a set of 100 queries, for which we report the average exact search time.

\begin{figure*}[tb]
	\centering
	\includegraphics[trim={1cm 10cm 2cm 3cm},scale=0.63]{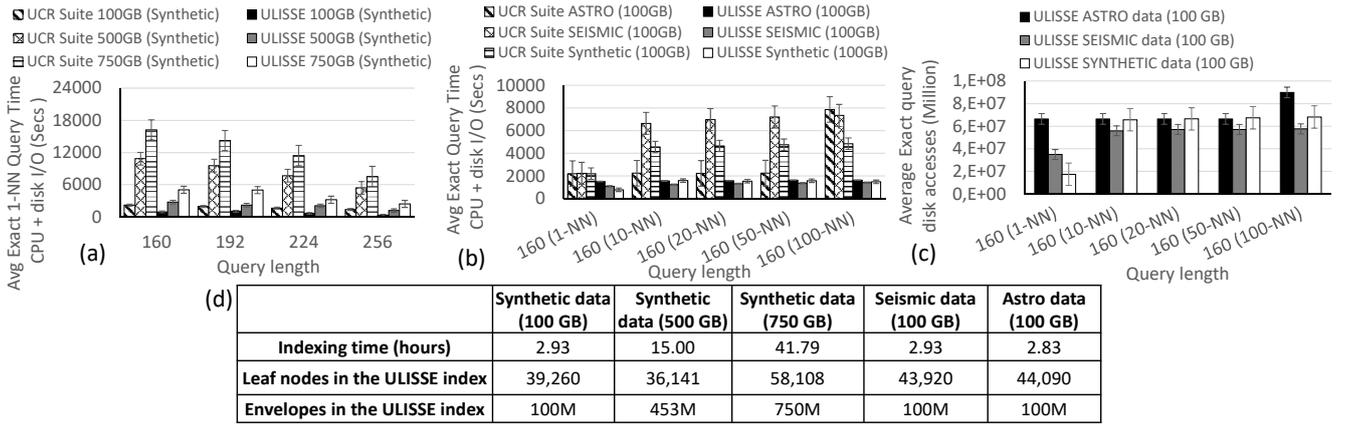}
	\caption{Exact and Approximate similarity search on Z-normalized synthetic and real datasets.
		\textit{a)} Average exact query time (CPU + disk I/O) on synthetic datasets. \textit{b)} Average exact \emph{k-NN} query time (CPU + disk I/O) on real datasets (100 GB) varying $k$. \textit{c)}  Average disk accesses of \emph{k-NN} query. \textit{d)} Indexing measures for all datasets.}
	\label{experimentsBig}
\end{figure*}

In Figure~\ref{experimentsBig}(a) we report the average query answering time (\emph{1-NN}) on synthetic datasets, varying the query length. These results demonstrate that $\textit{ULISSE}$ scales better than \textit{UCR Suite} across all query lengths, being up to 5x faster.

In Figure~\ref{experimentsBig}(b), we report the \emph{k-NN} exact search time performance, varying $k$ and picking the smallest query length, namely \textit{160}. 
Note that, this is the largest search space we consider in these datasets, since each query has \textit{9.7} billion of possible candidates (subsequences of length \textit{160}).
The experimental results on real datasets confirm the superiority of $\textit{ULISSE}$, which scales with stable performance, also when increasing the number $k$ of nearest neighbors. 
Once again it is up to 5x faster than \textit{UCR Suite}, whose performance deteriorates as $k$ gets larger. 

In Figure~\ref{experimentsBig}(c) we report the number of disk accesses of the queries considered in Figure~\ref{experimentsBig}(b). Here, we are counting the number of times that we follow a pointer from an envelope to the raw data on disk, during the sequential scan in Algorithm~\ref{exactSearch}.
Note that the number of disk accesses is bounded by the total number of Envelopes, which are reported in Figure~\ref{experimentsBig}(d) (along with the number of leaves and the building time for each index).

We observe that in the worst case, which takes place for the ASTRO dataset for $k=100$, we retrieve from disk $\sim$\textit{82}\% of the total number of subsequences. This still guarantees a remarkable speed-up over \textit{UCR Suite}, which needs to consider all the raw series.

Moreover, since $\textit{ULISSE}$ can use Early Abandoning during exact query answering, we observe during our empirical evaluation that disposing of the approximate answer distance prior the start of the exact search, permits to abandon on average \textit{20\%} of points more than \textit{UCR Suite} for the same query. 

\noindent{\bf Query over Large datasets with DTW.}
We conclude this part of the evaluation reporting the results of query answering on large datasets using the DTW distance.

In Figure~\ref{experimentsBigdtw}, we report the time performance of (\emph{1-NN} search) on the ASTRO, SEISMIC and synthetic datasets, each one containing $100M$ data series of length $256$ ($100$GB).
Also in this case, $\textit{ULISSE}$ guarantees a consistent speed-up over \textit{UCR Suite}, which is at least $\sim1.5$x faster in the worst case (ASTRO dataset, query length $160$), and up to one order of magnitude faster (synthetic dataset, query length $256$).

\begin{figure}[tb]
	\centering
	\includegraphics[trim={0cm 12cm 15cm 3cm},scale=0.80]{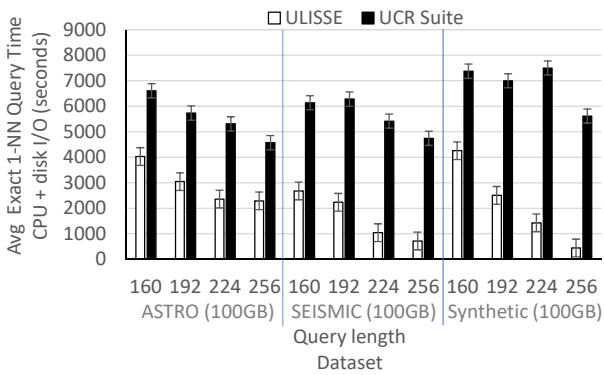}
	\caption{Average exact query time with DTW distance (CPU + disk I/O) on real and synthetic datasets.}
	\label{experimentsBigdtw}
\end{figure}

\subsection{\textit{ULISSE} vs Index Interpolation}

\begin{figure*}[tb]
	\centering
	\hspace*{-0.4cm}
	\includegraphics[trim={0cm 3cm 2cm 3cm},scale=0.65]{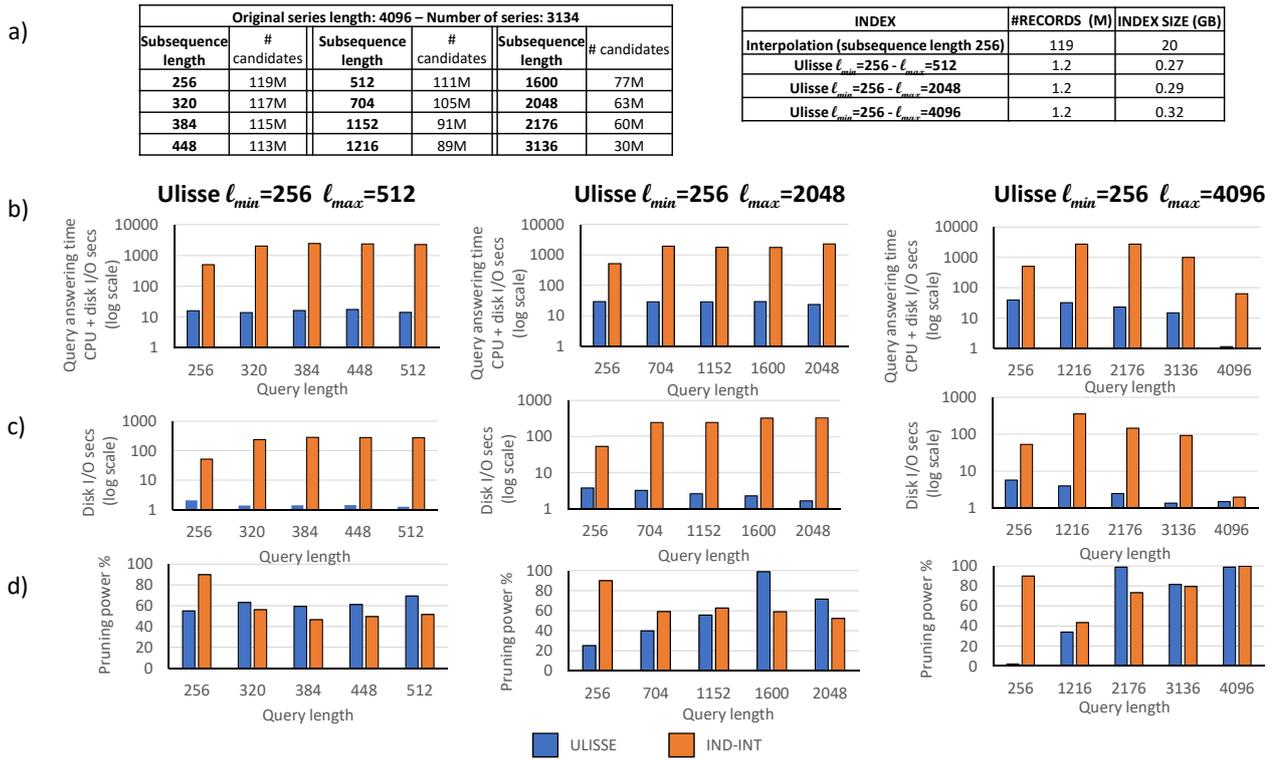}
	\caption{Result of $k$-NN search comparison between \textit{ULISSE} and \textit{IND-INT}. (\textit{a}) Data and Indices properties. (\textit{b}) Query answering time. (\textit{c}) Disk I/O. (\textit{d}) Pruning power \%.}
	\label{interUlisse}
\end{figure*}

In this section, we compare ULISSE to the Index Interpolation (\textit{IND-INT}) method.
\textit{IND-INT} works by means of $\epsilon$-range query answering of fixed length that serves the answering of variable length $k$-NN queries.
We consider the real datasets previously introduced and each query is answered using $k=1$ (Nearest Neighbor).
We set the other parameters to their default values.
The data series collections contain raw data series of fixed length $4096$. 
Hence, the number of candidates changes according to the (variable) length of the query subsequence.
We considered queries of lengths between $256$-$4096$.
In the left part of Figure~\ref{interUlisse}(a), we report the total number of candidate answers for each of these lengths.

In order to test the scalability of the approaches as a function of the query length range, we test three different ranges: ($256$-$512$), ($256$-$2048$), and ($256$-$4096$).
\textit{IND-INT} must build an index using the smallest query length ($256$), extracting one record for each candidate. 
We thus have the same index for all three ranges.
In contrast, for each one of the above three query length ranges, we can build a different \textit{ULISSE} index, whose sizes are reported in the right part of Figure~\ref{interUlisse}(a).
Observe that all three \textit{ULISSE} indexes have the same number of records (Envelopes), and that their sizes are two orders of magnitude smaller than the \textit{IND-INT} index.

In Figure~\ref{interUlisse}(b), we report the total query answering time (CPU + disk I/O; y-axis in log scale), as we vary the query length in the three chosen ranges.
\textit{ULISSE} answers $1$-NN queries more than $10x$ faster than \textit{IND-INT} in all these cases.
We observe the same in Figure~\ref{interUlisse}(c), where we report only the disk I/O time. 
Recall that \textit{ULISSE} starts by performing an approximate search that visits first the most promising nodes of the index. 
In contrast, \textit{IND-INT} issues an $\epsilon$-range query, and thus must probe each summarized record in the index, 
and then access the disk, when the lower bounding distance between the query and the record is smaller than $\epsilon$. 
This operation translates to significant time cost.

Note that in this set of experiments, we use for \textit{IND-INT} an $\epsilon$ equal to the distance between the query and its approximate answer, which we obtain by first running the query using \textit{ULISSE}.
This method for choosing $\epsilon$ (proposed in the \emph{RangeTopK} algorithm~\cite{DBLP:conf/vldb/HanLMJ07}) favors \textit{IND-INT}, since it provides a good initial value.

Our experiments show that \textit{ULISSE} scales better as the query length increases. 
As we observe in Figure~\ref{interUlisse}(d), \textit{ULISSE} always prunes more records as the query length increases, whereas \textit{IND-INT} exhibits an unstable pruning pattern that depends on the query length. 

Overall, \textit{ULISSE} uses a more succinct index that permits to scale better (since there are less records to iterate over). 
This translates to reduced CPU time, as well as disk accesses.

\subsection{$\epsilon$-Range Queries}
In this last part, we test the $\textit{ULISSE}$ search algorithm for the $\epsilon$-Range query task. 
To that extent we adapted Algorithm~\ref{exactSearch}, so that given as input $\epsilon \in \mathbb{R}$, it computes the set of subsequences that have a distance to the query smaller than or equal to $\epsilon$. 
Similarly, we also adapted the \textit{UCR Suite} to support $\epsilon$-Range search.
As additional competitors, we consider \textit{IND-INT} (only for Euclidean distance), and \textit{KV-Match}, which is the state-of-the art index-based solution for exact $\epsilon$-Range queries on non Z-normalized data series.

In this experiment, we used five different real datasets, composed by a single data series of different lengths, as reported in Figure~\ref{epsRangeQuery}(a).
For each of these datasets, we can see that $\textit{ULISSE}$ builds its index $5$ times faster than \textit{KV-Match}. 
This is because \textit{KV-Match} is based on the construction of multiple indexes. 
Specifically, it builds different indexes performing a sliding windows extraction at different lengths.
At query answering time, \textit{KV-Match} performs a recombination of query answers coming from the different indexes. 

For our $\epsilon$-Range queries, we set the $\epsilon$ parameter to twice the \emph{NN} distance of each query. 
In this manner, we simulate an exploratory analysis task.   
We report the average value of query selectivity in Figure~\ref{epsRangeQuery}(b).
We note that in the ECG dataset the selectivity is very high. This is due to the periodic/cyclical nature of this kind of data, which contain repeating heartbeats subsequences that are very similar.
In the other datasets, we have different values of selectivity ranging from $0.5\%$ to $15\%$, when using Euclidean distance.
On the other hand, when the DTW measure is considered, we observe a significant increase of the answer-set cardinality. 

In Figures~\ref{epsRangeQuery}(c) and~(d), we show the average query answering time for Euclidean distance, when varying the query length and the dataset, respectively. 
The results show that \textit{IND-INT} does not represent a competitive alternative.
In fact, only in one case, when the number of candidates is the smallest (i.e., $256$), it has time performance better than the other approaches.
When the number of possible answers increases (smaller query lengths), we observe that the performance of \textit{IND-INT} becomes more than an order of magnitude worse than the rest.

We note that in this case \textit{$ULISSE$} and \textit{KV-match} have no substantial difference in their time performance, with \textit{$ULISSE$} being slightly better.

However, when we consider the DTW distance, \textit{$ULISSE$} becomes up to one order of magnitude faster than \textit{KV-Match}, as shown in Figures~\ref{epsRangeQuery}(e) and~(f).
This difference is pronounced for the two largest datasets: \textit{$ULISSE$} is 3x faster for ECG, and 10x faster for GAP. 
It is important to note that since \textit{KV-Match} needs to recombine the answers from the different index structures, its time performance is affected by this refinement phase of query answering, and is rather sensitive to dataset size and query selectivity.

\begin{figure*}[tb]
	\centering
	\includegraphics[trim={0cm 9cm 2cm 3cm},scale=0.60]{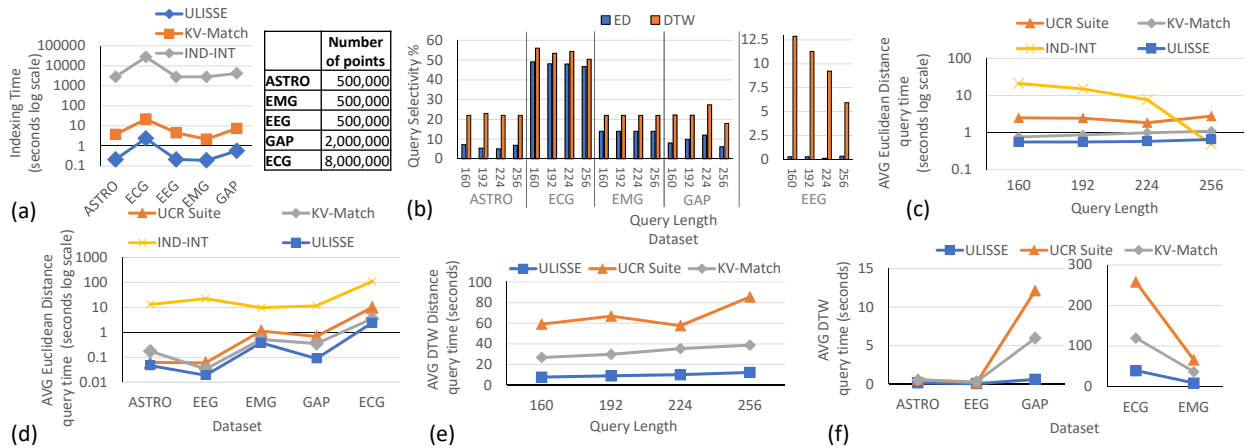}
	\caption{Results of $\epsilon$-range search on non Z-normalized real datasets. \textit{(a)} Indexing time and datasets length. \textit{(b)} Average selectivity of the queries in each dataset. \textit{(c)} Average query answering time for each query length, using Euclidean distance. \textit{(d)} Average query answering time for each dataset, using Euclidean distance. \textit{(d)} Average query answering time for each query length, using DTW. \textit{(e)} Average query answering time for each dataset, using DTW.}
	\label{epsRangeQuery}
\end{figure*}

\section{Conclusions}
\label{sec:conclusions}

Similarity search is one of the fundamental operations for several data series analysis tasks.
Even though much effort has been dedicated to the development of indexing techniques that can speed up similarity search, all existing solutions are limited by the fact that they can only support queries of a fixed length. 

In this work, we proposed \textit{ULISSE}, the first index able to answer similarity search queries of variable-length, over both Z-normalized and non Z-normalized sequences, supporting the Euclidean and DTW distances, for answering exactly, or approximately both \emph{k-NN} and $\epsilon$-range queries.
We experimentally evaluated, our indexing and similarity search algorithms, on synthetic and real datasets, demonstrating the effectiveness and efficiency (in space and time cost) of the proposed solution.




\bibliographystyle{abbrv}       
\bibliography{ulisse}   


\end{document}